\documentclass[twocolumn,11pt]{article}

\usepackage[utf8]{inputenc}
\usepackage[english]{babel}
\usepackage[T1]{fontenc}
\usepackage{amsmath,amssymb,amsthm}
\usepackage{mathtools}
\usepackage{hyperref}
\usepackage{float}
\usepackage[mathscr]{euscript}

\usepackage{mathrsfs}
\usepackage{dsfont}
\usepackage{setspace}
\usepackage[export]{adjustbox}
\usepackage{makecell}

\usepackage{enumitem}
\usepackage{color}

\usepackage{multirow}

\usepackage{tikz}
\usetikzlibrary{shapes.geometric,arrows.meta, chains,positioning,patterns,decorations.pathreplacing}
\usepackage{subcaption}

\usepackage{algorithm}
\usepackage{algpseudocodex}
\usepackage{setspace}

\newtheorem{theorem}{Theorem}
\newtheorem*{theorem*}{Theorem}
\newtheorem{lemma}[theorem]{Lemma}

\newtheorem{claim}[theorem]{Claim}

\usepackage{thmtools}
\usepackage{thm-restate}

\usepackage[sort&compress,numbers]{natbib}

\usepackage{mystyle}

\usepackage{authblk}

\title{Addressing stopping failures for small set flip decoding of hypergraph product codes}
\author[1]{Lev Stambler}
\affil[1]{University of Maryland, MD, 20742, USA}
\author[2]{Anirudh Krishna}
\affil[2]{Department of Computer Science \& Stanford Institute for Theoretical Physics, Stanford, CA, 94305, USA}
\author[3]{Michael E. Beverland}
\affil[3]{Microsoft Quantum, Redmond, WA 98052, USA}

\begin{document}

\twocolumn[
\begin{@twocolumnfalse}
\maketitle

\begin{abstract}
  For a quantum error correcting code to be used in practice, it needs to be equipped with an efficient decoding algorithm, which identifies corrections given the observed syndrome of errors.
  Hypergraph product codes are a promising family of constant-rate quantum LDPC codes that have a linear-time decoding algorithm called Small-Set-Flip ($\ssf$) (Leverrier, Tillich, Z\'emor FOCS 2015).
  The algorithm proceeds by iteratively applying small corrections which reduce the syndrome weight.
  Together, these small corrections can provably correct large errors for sufficiently large codes with sufficiently large (but constant) stabilizer weight.
  However, this guarantee does not hold for small codes with low stabilizer weight.
  In this case, $\ssf$ can terminate with stopping failures, meaning it encounters an error for which it is unable to identify a small correction.
  We find that the structure of errors that cause stopping failures have a simple form for sufficiently small qubit failure rates.
  We propose a new decoding algorithm called the Projection-Along-a-Line ($\pal$) decoder to supplement $\ssf$ after stopping failures.
  Using $\ssf+\pal$ as a combined decoder, we find an order-of-magnitude improvement in the logical error rate.
\end{abstract}

\date{}
\vspace{0.5cm}
\end{@twocolumnfalse}]


\section{Motivation and main results}
\vspace{-0.5cm}
Finite-rate quantum low-density parity-check (LDPC) codes are compelling candidates for building quantum memories.
The \emph{low density} aspect refers to properties that enable syndrome extraction for these codes to involve each qubit in only a constant number of two-qubit gates, independent of the code size.
This limits the spreading of errors which helps ensure that logical errors can be suppressed exponentially as a function of the size of the code.
Simultaneously, by virtue of their \emph{finite rate}, the number of physical qubits needed to encode each logical qubit in these codes is independent of the code's size.
This offers a considerable benefit over the popular geometrically local code families such as surface codes and color codes which have vanishing rate \cite{kovalev2012improved,gottesman2014fault,bravyi2010tradeoffs,baspin2022connectivity,baspin2022quantifying,baspin2023improved}.

There is a plethora of constructions of constant-rate quantum LDPC codes, including those based on hyperbolic manifolds \cite{freedman2002z2,guth2014quantum,londe2017golden}, the hypergraph product \cite{tillich2014quantum,zeng2019higher}, high-dimensional expanders \cite{evra2020decodable,kaufman2021new} and algebraic topology \cite{hastings2021fiber,breuckmann2021balanced,lin2022good}.
Recently, a number of long sought-after `good' quantum codes have been discovered \cite{panteleev2022asymptotically,leverrier2022quantum}, for which both the number of logical qubits and the distance scale linearly with the number of physical qubits.
In this work, we focus on the hypergraph product code (HGP) construction, which uses pairs of constant-rate \emph{classical} LDPC codes to build constant-rate \emph{quantum} LDPC codes.
While the asymptotic properties of HGP codes fall short of some of the newer families, the HGP construction is arguably the simplest method to build code families with practically reasonable parameters and is the starting point for many of the other constructions.

\begin{figure*}[t]
    \centering
    \begin{tikzpicture}

\end{tikzpicture}
    
\begin{tikzpicture}
    






    \begin{scope}[scale=0.5,xshift=175]
    \fill[pattern=north west lines] (3.2,5.7) rectangle (3.7,6.2);
    \fill[pattern=north west lines] (5.2,5.2) rectangle (5.7,5.7);
    \fill[pattern=north west lines] (3.7,3.7) rectangle (4.2,4.2);
    \fill[pattern=north west lines] (4.2,4.7) rectangle (4.7,5.2);
    \fill[pattern=north west lines] (4.7,3.7) rectangle (5.2,4.2);
    \draw (0,0) rectangle (3,3);
    \node at (-1.5,1.5) {$C_1$};
    \draw (0,3.2) rectangle (3,6.2);
    \node at (-1.5,4.7) {$V_1$};
    \draw (3.2,0) rectangle (6.2,3);
    \node at (1.5,-1.5) {$V_2$};
    \draw (3.2,3.2) rectangle (6.2,6.2);
    \node at (4.7,-1.5) {$C_2$};
    
    \foreach \x in {3.2,3.7,4.2,4.7,5.2,5.7} {
    \foreach \y in {3.2,3.7,4.2,4.7,5.2,5.7} {
        \draw ({\x},\y) rectangle ({\x+0.5},{\y+0.5});
     }
    }
    \foreach \x in {0,0.5,1,1.5,2,2.5} {
    \foreach \y in {3.2,3.7,4.2,4.7,5.2,5.7} {
        \draw ({\x},\y) rectangle ({\x+0.5},{\y+0.5});
     }
    }
    \foreach \x in {3.2,3.7,4.2,4.7,5.2,5.7} {
    \foreach \y in {0,0.5,1,1.5,2,2.5} {
        \draw ({\x},\y) rectangle ({\x+0.5},{\y+0.5});
     }
    }
    \foreach \x in {0,0.5,1,1.5,2,2.5} {
    \foreach \y in {0,0.5,1,1.5,2,2.5} {
        \draw ({\x},\y) rectangle ({\x+0.5},{\y+0.5});
     }
    }

    \foreach \y in {0,0.5,1,1.5,2,2.5} {
        \draw (-1,\y) rectangle (-0.5,{\y+0.5});
     }

     \foreach \y in {3.2,3.7,4.2,4.7,5.2,5.7} {
        \draw (-1,\y) rectangle (-0.5,{\y+0.5});
     }

     \foreach \x in {0,0.5,1,1.5,2,2.5} {
        \draw (\x,-0.5) rectangle ({\x+0.5},-1);
     }

     \foreach \x in {3.2,3.7,4.2,4.7,5.2,5.7} {
        \draw (\x,-0.5) rectangle ({\x+0.5},-1);
     }

    \node at (3,-2.5) {After random $Z$ error};
     \node at (3,-3.25) {(unstructured syndrome)};
     \draw[->] (7,3) --node[above] {$\ssf$} (10,3);
     \end{scope}
    \begin{scope}[scale=0.5,xshift=550]
    \foreach \x in {3.7,4.2,4.7,5.2} {
        \fill[pattern=north west lines] (\x,3.7) rectangle ({\x+0.5},4.2);
     }
    \draw (0,0) rectangle (3,3);
    \node at (-1.5,1.5) {$C_1$};
    \draw (0,3.2) rectangle (3,6.2);
    \node at (-1.5,4.7) {$V_1$};
    \draw (3.2,0) rectangle (6.2,3);
    \node at (1.5,-1.5) {$V_2$};
    \draw (3.2,3.2) rectangle (6.2,6.2);
    \node at (4.7,-1.5) {$C_2$};
    
    \foreach \x in {3.2,3.7,4.2,4.7,5.2,5.7} {
    \foreach \y in {3.2,3.7,4.2,4.7,5.2,5.7} {
        \draw ({\x},\y) rectangle ({\x+0.5},{\y+0.5});
     }
    }
    \foreach \x in {0,0.5,1,1.5,2,2.5} {
    \foreach \y in {3.2,3.7,4.2,4.7,5.2,5.7} {
        \draw ({\x},\y) rectangle ({\x+0.5},{\y+0.5});
     }
    }
    \foreach \x in {3.2,3.7,4.2,4.7,5.2,5.7} {
    \foreach \y in {0,0.5,1,1.5,2,2.5} {
        \draw ({\x},\y) rectangle ({\x+0.5},{\y+0.5});
     }
    }
    \foreach \x in {0,0.5,1,1.5,2,2.5} {
    \foreach \y in {0,0.5,1,1.5,2,2.5} {
        \draw ({\x},\y) rectangle ({\x+0.5},{\y+0.5});
     }
    }

    \foreach \y in {0,0.5,1,1.5,2,2.5} {
        \draw (-1,\y) rectangle (-0.5,{\y+0.5});
     }

     \foreach \y in {3.2,3.7,4.2,4.7,5.2,5.7} {
        \draw (-1,\y) rectangle (-0.5,{\y+0.5});
     }

     \foreach \x in {0,0.5,1,1.5,2,2.5} {
        \draw (\x,-0.5) rectangle ({\x+0.5},-1);
     }

     \foreach \x in {3.2,3.7,4.2,4.7,5.2,5.7} {
        \draw (\x,-0.5) rectangle ({\x+0.5},-1);
     }

     \node at (3,-2.5) {After $\ssf$ halt failure};
     \node at (3,-3.25) {(structured syndrome)};
     \draw[->] (7,3) --node[above] {$\pal$} (10,3);
     \end{scope}
     \begin{scope}[scale=0.5,xshift=925]
    \draw (0,0) rectangle (3,3);
    \node at (-1.5,1.5) {$C_1$};
    \draw (0,3.2) rectangle (3,6.2);
    \node at (-1.5,4.7) {$V_1$};
    \draw (3.2,0) rectangle (6.2,3);
    \node at (1.5,-1.5) {$V_2$};
    \draw (3.2,3.2) rectangle (6.2,6.2);
    \node at (4.7,-1.5) {$C_2$};
    
    \foreach \x in {3.2,3.7,4.2,4.7,5.2,5.7} {
    \foreach \y in {3.2,3.7,4.2,4.7,5.2,5.7} {
        \draw ({\x},\y) rectangle ({\x+0.5},{\y+0.5});
     }
    }
    \foreach \x in {0,0.5,1,1.5,2,2.5} {
    \foreach \y in {3.2,3.7,4.2,4.7,5.2,5.7} {
        \draw ({\x},\y) rectangle ({\x+0.5},{\y+0.5});
     }
    }
    \foreach \x in {3.2,3.7,4.2,4.7,5.2,5.7} {
    \foreach \y in {0,0.5,1,1.5,2,2.5} {
        \draw ({\x},\y) rectangle ({\x+0.5},{\y+0.5});
     }
    }
    \foreach \x in {0,0.5,1,1.5,2,2.5} {
    \foreach \y in {0,0.5,1,1.5,2,2.5} {
        \draw ({\x},\y) rectangle ({\x+0.5},{\y+0.5});
     }
    }

    \foreach \y in {0,0.5,1,1.5,2,2.5} {
        \draw (-1,\y) rectangle (-0.5,{\y+0.5});
     }

     \foreach \y in {3.2,3.7,4.2,4.7,5.2,5.7} {
        \draw (-1,\y) rectangle (-0.5,{\y+0.5});
     }

     \foreach \x in {0,0.5,1,1.5,2,2.5} {
        \draw (\x,-0.5) rectangle ({\x+0.5},-1);
     }

     \foreach \x in {3.2,3.7,4.2,4.7,5.2,5.7} {
        \draw (\x,-0.5) rectangle ({\x+0.5},-1);
     }

     \node at (3,-2.5) {After $\pal$ decoder};
     \node at (3,-3.25) {(trivial syndrome)};
     \end{scope}
\end{tikzpicture}
    \caption{
    Two-stage decoding of HGP codes.
    The HGP code has a Tanner graph $\cG$ formed from the Cartesian product of the Tanner graphs of two classical codes with bit nodes $C_1, C_2$ and check nodes $V_1, V_2$.
    Vertices of $\cG$ are represented by small squares, with qubits associated with $V_1 \times V_2$ and $C_1 \times C_2$ vertices, while $Z$ and $X$ type stabilizer generators are associated with $C_1 \times V_2$ and $V_1 \times C_2$ vertices respectively.
    In the first stage, the $\ssf$ decoder is applied, and if a valid correction is output then decoding is finished.
    If $\ssf$ experiences a stopping failure (as in the example depicted), we apply a second decoding stage using the \emph{Projection Along a Line} decoder ($\pal$).
    Our design of $\pal$ exploits the structure of the stopping syndrome which remains after an $\ssf$ stopping failure, which inherits properties of the classical codes.}
    \label{fig:decoding-schematic}
\end{figure*}
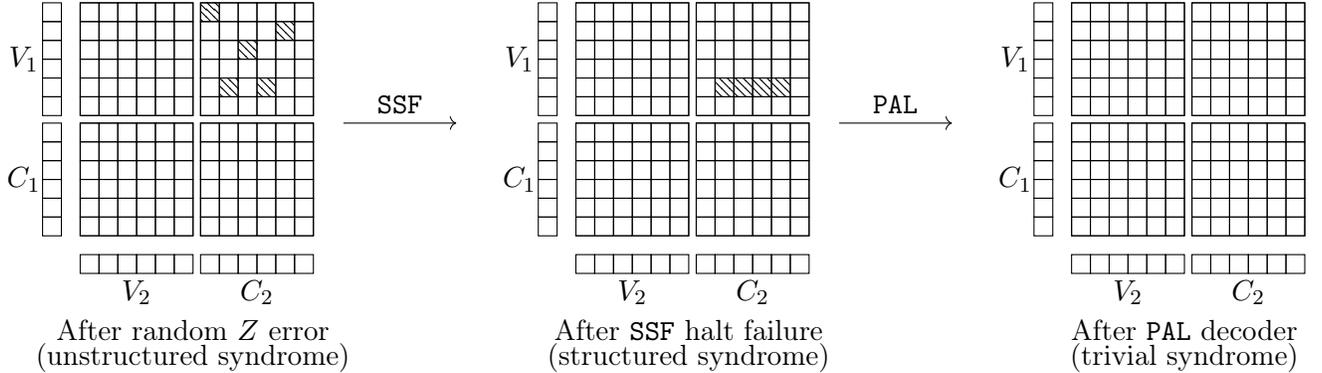

\emph{Decoding algorithms} play a crucial role in quantum error correction. 
These classical algorithms identify a correction from the information output by syndrome extraction.
Arguably the most famous example of a decoding algorithm for quantum LDPC codes is \ttt{matching} for surface codes \cite{kitaev1997quantum,bravyi1998quantum}, which is based on Edmonds' graph matching algorithm.
This algorithm relies on geometric locality, as do many other decoding algorithms that apply to various geometrically-local codes such as the renormalization-group decoder for topological codes \cite{duclos2010fast} and the projection decoder for color codes \cite{delfosse2014decoding}.
The Union Find decoder \cite{delfosse2021almost} notably applies to many codes including non-geometrically local ones \cite{delfosse2021union,delfosse2022toward}, but is only known to have a threshold for geometrically-local codes.
However, constant-rate quantum LDPC codes are not geometrically local~\cite{bravyi2010tradeoffs,baspin2022connectivity,baspin2022quantifying,baspin2023improved}, necessitating different properties such as graph expansion to form the basis of their decoders.

We focus our attention on Small-Set-Flip ($\ssf$), which is a linear-time decoding algorithm for HGP codes introduced by Leverrier, Tillich and Z\'emor~\cite{leverrier2015quantum}.
With adversarial Pauli errors, 
$\ssf$ has been proven to efficiently correct errors of weight up to a constant-factor of the code distance~\cite{leverrier2015quantum}.
With stochastic Pauli errors with perfect syndrome measurements, which is the noise model we consider in this work, $\ssf$ has a finite threshold~\cite{fawzi2018efficient}.
Moreover, a finite threshold still exists when stochastic syndrome errors are added (in which case decoding can be performed in a single-shot manner \cite{fawzi2018constant}).
For these reasons, $\ssf$ seems very promising for practical applications, motivating further improvements and optimizations to be considered, particularly in parameter regimes of small code sizes and stabilizer weights that are relevant for use in quantum memories.

In this work, we present improvements to the Small-Set-Flip ($\ssf$) \cite{leverrier2015quantum} decoding algorithm for HGP codes.
Our primary contribution is to address error configurations for which $\ssf$ cannot find a correction matching the observed syndrome, which we call a \emph{stopping} failure.
The $\ssf$ algorithm runs iteratively, reducing the syndrome weight in a sequence of rounds by applying partial corrections. 
Typically when $\ssf$ terminates in a stopping failure, it has already undergone many rounds and has made partial progress in correcting the error.
Moreover, unlike the syndrome of the initial stochastic error introduced by the noise model, we find that the \emph{stopping syndrome} remaining after $\ssf$ halts tends to be very structured.
We exploit this structure to design the Projection Along a Line ($\pal$) decoder that can typically correct the stopping set output by $\ssf$ when it results in stopping failures; see \fig{decoding-schematic}.

We note that generalizations of $\ssf$ \cite{leverrier2022efficient,leverrier2022quantum,leverrier2022parallel,gu2022efficient} have been used to design decoders for other families of quantum LDPC codes, including those which have both constant rate and constant relative distance~\cite{panteleev2022asymptotically, leverrier2022quantum}.
Consequently, we believe that many of the proposals described in this paper can generalized to improve decoders for other codes as well.

The structure of the remainder of this paper is as follows.
In \sec{background}, we set notation and review relevant background on error correction, decoders, HGP codes and specify how we perform numerical performance estimates. 
In \sec{pal}, we analyze stopping failures of the $\ssf$ decoder, motivate and define the $\pal$ decoding algorithm and analyze its performance and time complexity.
For instance, we find that for a HGP code with 3600 physical qubits subjected to 1\% physical error rate with perfect check measurements, the logical error rate falls from roughly $5 \times 10^{-2}$ to $4 \times 10^{-4}$ when moving from the $\ssf$ decoder to the $\ssf$ decoder followed by $\pal$.
We also discuss a number of approaches to improve the implementation run time of $\ssf$ in \app{space-time-cost}.

\paragraph{Related work:}
Combining $\ssf$ with other decoding algorithms has been explored before.
In \cite{Grospellier2021}, $\ssf$ was applied following Belief Propagation ($\bp$); the combination was referred to as $\bp+\ssf$.

Connolly \emph{et al}.\ \cite{connolly2022fast} design a decoding algorithm that also exploits the product structure of HGP codes.
Their decoding algorithm reduces the problem of decoding erasure errors on HGP codes to that of decoding erasure errors on a classical code.
As a result, their erasure decoding algorithm has complexity $O(N^2)$ in contrast to $O(N^3)$ for generic codes.
In Ref.~\cite{connolly2022fast}, the authors propose one way to apply their approach to the setting of stochastic Pauli noise, namely by using it as a subroutine for the Union Find decoder~\cite{delfosse2021almost}.
However, the Union Find decoder is not known to have a threshold for HGP codes.
Our work can be seen as another approach to generalize the work of Connolly \emph{et al}.\ to the case of stochastic Pauli noise, but in combination with the SSF decoder, which is known to have a threshold for HGP codes.

\section{Background and notation}
\label{sec:background}

In this section we review some important background and fix the notation we use in the rest of the paper.

\subsection{Error correction}
A classical $[n,k,d]$ error correcting code $\cC \subseteq \bbF_2^n$ is a $k$-dimensional subspace such that every element has Hamming weight at least $d$.
A code $\cC$ can be specified by a parity check matrix $\pcm \in \bbF_2^{m \times n}$ ---the code is the kernel of $\pcm$.
Given an error $\cerr \in \bbF_2^n$, we let $\synd = \pcm \cerr$ denote its syndrome.
A family of codes $\{\cC_n\}$ is said to be a low-density parity check (LDPC) family if each row and column of the parity check matrix for each instance has weight upper bounded by a constant independent of $n$.
Equivalently, the code can be specified by its Tanner graph $G = (V \sqcup C,E)$, a bipartite graph where each node $v \in V$ represents a bit and each node $c \in C$ represents a parity check;
$(v,c) \in E$ if the bit $v$ is in the support of the check $c$.

Given a Tanner graph $G = (V \sqcup C,E)$,
for two sets $S \subseteq V$ and $T \subseteq C$, we let $E(S,T)$ denote the edges between $S$ and $T$.
For $S \subseteq V$ and $T \subseteq C$, we let $\Gamma(S)$ and $\Gamma(T)$ denote the neighborhood, i.e.\ the set of nodes adjacent to $S$ and $T$:
\begin{align*}
    \Gamma(S) &= \set{c | E(S,\{c\})  \neq \emptyset }~,\\
    \Gamma(T) &= \set{v | E(\{v\},T) \neq \emptyset }~.
\end{align*}
We let $\tC$ represent the $[m,\tk,\td]$ code obtained from $G$ by exchanging the role of the bits and parity checks, which we call the \emph{dual code}.

A \emph{stabilizer group} over $N$ qubits is an Abelian subgroup $\cS \subseteq \cP_N$ of the $N$-qubit Pauli group and can be expressed using $M$ independent stabilizer generators
\begin{equation}
    \cS = \langle S_{i} : i \in \{1,...,M\} \rangle~.
\end{equation}
The stabilizer group defines a \emph{stabilizer code} $\cQ$, which is the space of joint $+1$-eigenstates of the generators,
\begin{equation}
    \cQ = \{ \ket{\psi} : S \ket{\psi} = \ket{\psi} \; \forall \; S \in \cS \}~.
\end{equation}
The (log-)dimension $K$ of this space is equal to $N-M$ and corresponds to the number of qubits that can be encoded.

The set of \emph{logical operators} $\cL$ is the subset of Pauli operators that commute with all the stabilizer generators,
\begin{equation}
    \cL = \{\ssL : [\ssL, \ssS] = 0 \; \forall \; S \in \cS \}~.
\end{equation}
The \emph{distance} $D$ of the stabilizer code is the weight of the smallest $L \in \cL$ such that $L$ is not in the stabilizer group $\cS$ up to a phase.

A family of stabilizer codes $\{\cQ_N\}$ is said to be a family of \emph{quantum LDPC codes} if, for all elements $\cQ_N$, each qubit is acted on by at most $\Delta_q$ stabilizer generators and each stabilizer generator acts on at most $\Delta_g$ qubits.
We then say that $\{\cQ_N\}$ is an $\dsl N,K(N),D(N), \Delta_q, \Delta_g\dsr$ quantum LDPC code family.
We say the family has \emph{finite rate} if $K(N) = \Theta(N)$.

The quantum codes we consider in this paper are CSS codes~\cite{calderbank1996good,steane1996multiple}, an important class of stabilizer codes for which each stabilizer generator is expressed as either a purely $X$ type or a purely $Z$ type Pauli operator.
This then means that the larger decoding task can be broken down into two separate tasks.
In this setting, it is convenient to represent the error as a pair of subsets of qubits corresponding to the $X$ and $Z$ support.
We focus on correcting only $Z$-type errors, and precisely the same approach can be used to correct $X$-type errors.
Moreover, both the codes and the noise that we consider are symmetric under the exchange of $X$ and $Z$, such that the details of any analysis of the correction of $Z$ errors also applies directly to the correction of $X$ errors.

Given a $Z$ error specified by its support $E \subseteq [N]$, the \emph{syndrome} $\synd(E)$ is a bit string that records the outcomes of a complete set of generators of the stabilizer group. 
Each bit in $\synd(E)$ corresponds to a generator $S \in \cS$, and is $0$ if $Z(E)$ and $S$ commute, and is $1$ otherwise.

\subsection{Hypergraph product codes}
\label{sec:hgp}

Let $G = (V \union C, E)$ be a bipartite graph.
We assume that $G$ is connected, i.e.\ there is a path from any vertex to any other vertex in $G$.
We assume $G$ is biregular, i.e.\ all vertices in $V$ and $C$ have degrees $\Delta_V$ and $\Delta_C$ respectively.

Let $G_1 = (V_1 \union C_1,E_1)$ and $G_2 = (V_2 \union C_2, E_2)$ represent two copies of $G$.
Let $\cG = G_1 \times G_2$ denote the Cartesian product of $G$ with itself.
The quantum code is defined with respect to $\cG$ as follows:
\begin{itemize}
\item The set of physical qubits is associated with the set $V_1 \times V_2 \sqcup C_1 \times C_2$.

\item The set of $Z$ type stabilizer generators, denoted $\cZ$, is associated with the set $C_1 \times V_2$.
Therefore given $c_1 \in C_1$ and $v_2 \in V_2$, generator $(c_1,v_2)$ is supported on qubits $\Gamma(c_1) \times \{v_2\} \sqcup \{c_1\} \times \Gamma(v_2)$.
This set is called the \emph{support} of the generator $(c_1,v_2)$ and denoted $\supp(c_1,v_2)$.

\item The set of $X$ type stabilizer generators, denoted $\cX$, is associated with the set $V_1 \times C_2$.
Therefore given $v_1 \in V_1$ and  $c_2 \in C_2$, generator $(v_1,c_2)$ is supported on qubits $\{v_1\} \times \Gamma(c_2) \sqcup \Gamma(v_1) \times \{c_2\}$.
\end{itemize}

It is straightforward to see that the stabilizers commute and therefore define a valid quantum code.
Consider a $Z$ type stabilizer generator $(c_1,v_2)$, which has support on qubits $\Gamma(c_1) \times \{v_2\} \sqcup \{c_1\} \times \Gamma(v_2)$.
$(c_1,v_2)$ shares support with the set of $X$ type stabilizer generators $\Gamma(c_1) \times \Gamma(v_2)$.
This set of parity checks is also called the \emph{local view} of the generator $(c_1,v_2)$.

Take any one of these $X$ type stabilizer generators, say $(v_1,c_2)$ where $v_1 \in \Gamma(c_1)$ and $c_2 \in \Gamma(v_2)$, and note that this $X$ type stabilizer generator is supported on $\{v_1\} \times \Gamma(c_2) \sqcup \Gamma(v_1) \times \{c_2\}$. 
$(v,c)$ and $(c_1,v_2)$ therefore share support on precisely two qubits: $(v_1,v_2)$ and $(c_1,c_2)$ and
consequently commute.
A schematic for the code is presented in \fig{exp-code}.

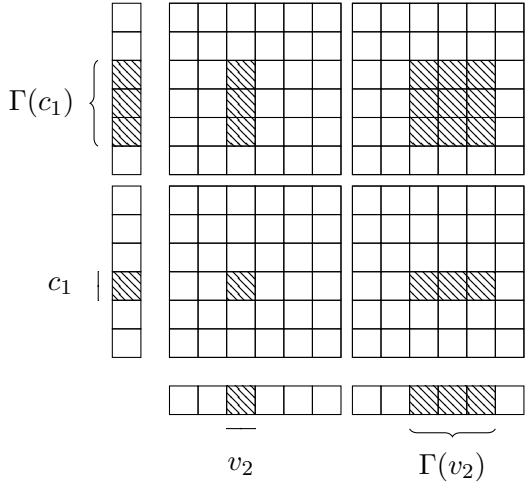
\begin{figure}[h]
    \centering
     \begin{tikzpicture}
        \begin{scope}[scale=0.76]
        \fill[pattern=north west lines] (-1,1) rectangle (-0.5,1.5);
        \foreach \y in {3.7,4.2,4.7} {
            \fill[pattern=north west lines] (-1,\y) rectangle (-0.5,{\y+0.5});
        }
        \fill[pattern=north west lines] (1,-1) rectangle (1.5,-0.5);
        \foreach \x in {4.2,4.7,5.2} {
            \fill[pattern=north west lines] (\x,-1) rectangle ({\x+0.5},-0.5);
        }
        \fill[pattern=north west lines] (1,1) rectangle (1.5,1.5);
        \foreach \y in {3.7,4.2,4.7} {
            \fill[pattern=north west lines] (1,\y) rectangle (1.5,{\y+0.5});
        }
        \foreach \x in {4.2,4.7,5.2} {
            \fill[pattern=north west lines] (\x,1) rectangle ({\x+0.5},1.5);
        }
        \foreach \x in {4.2,4.7,5.2} {
            \foreach \y in {3.7,4.2,4.7} {
                \fill[pattern=north west lines] (\x,\y) rectangle ({\x+0.5},{\y+0.5});
            }
        }
        
        \draw (0,0) rectangle (3,3);
        \draw (0,3.2) rectangle (3,6.2);
        \draw [decorate,decoration={brace,amplitude=3pt},yshift=0pt]
(-1.25,3.7) -- (-1.25,5.2) node [black,midway,xshift=-0.75cm] {$\Gamma(c_1)$};
        \draw (3.2,0) rectangle (6.2,3);
        \draw [decorate,decoration={brace,mirror,amplitude=-0.15pt},yshift=0pt]
(-1.25,1.5) -- (-1.25,1) node [black,midway,xshift=-0.5cm] {$c_1$};
        \draw (3.2,3.2) rectangle (6.2,6.2);
        \draw [decorate,decoration={brace,amplitude=-0.15pt},yshift=0pt]
(1.5,-1.25) -- (1,-1.25) node [black,midway,yshift=-0.5cm] {$v_2$};
        \draw [decorate,decoration={brace,amplitude=3pt,mirror},yshift=0pt]
(4.2,-1.25) -- (5.7,-1.25) node [black,midway,yshift=-0.5cm] {$\Gamma(v_2)$};
        
        \foreach \x in {3.2,3.7,4.2,4.7,5.2,5.7} {
        \foreach \y in {3.2,3.7,4.2,4.7,5.2,5.7} {
            \draw ({\x},\y) rectangle ({\x+0.5},{\y+0.5});
         }
        }
        \foreach \x in {0,0.5,1,1.5,2,2.5} {
        \foreach \y in {3.2,3.7,4.2,4.7,5.2,5.7} {
            \draw ({\x},\y) rectangle ({\x+0.5},{\y+0.5});
         }
        }
        \foreach \x in {3.2,3.7,4.2,4.7,5.2,5.7} {
        \foreach \y in {0,0.5,1,1.5,2,2.5} {
            \draw ({\x},\y) rectangle ({\x+0.5},{\y+0.5});
         }
        }
        \foreach \x in {0,0.5,1,1.5,2,2.5} {
        \foreach \y in {0,0.5,1,1.5,2,2.5} {
            \draw ({\x},\y) rectangle ({\x+0.5},{\y+0.5});
         }
        }

        \foreach \y in {0,0.5,1,1.5,2,2.5} {
            \draw (-1,\y) rectangle (-0.5,{\y+0.5});
         }

         \foreach \y in {3.2,3.7,4.2,4.7,5.2,5.7} {
            \draw (-1,\y) rectangle (-0.5,{\y+0.5});
         }

         \foreach \x in {0,0.5,1,1.5,2,2.5} {
            \draw (\x,-1) rectangle ({\x+0.5},-0.5);
         }

         \foreach \x in {3.2,3.7,4.2,4.7,5.2,5.7} {
            \draw (\x,-1) rectangle ({\x+0.5},-0.5);
         }
         \end{scope}
         \end{tikzpicture}
    \caption{Schematic for the support of a $Z$ type stabilizer generator $(c_1,v_2)$ for $c_1 \in C_1$ and $v_2 \in V_2$.}
    \label{fig:exp-code}
\end{figure}
The parameters of the HGP code are specified by the parameters of the codes $\cC$ and $\tC$ associated with $G$.
In particular, if $\cC$ or $\tC$ is constant rate, then so is the resulting quantum code.
Furthermore, if $G$ is a graph with bounded degree $\Delta_V$ and $\Delta_C$, then the quantum code is an LDPC code that obeys $\Delta_q \leq 2\max(\Delta_V,\Delta_C)$ and $\Delta_g \leq \Delta_V + \Delta_C$.

We see that HGP codes are CSS codes, and as such we consider only $Z$ type errors.
We refer to errors on $V_1 \times V_2$ as VV errors and errors on $C_1 \times C_2$ as CC errors.
We refer to $X$ type stabilizer generators as parity checks and $Z$ type stabilizer generators as generators.
If a parity check is adjacent to an odd number of (qu)bits with errors, then the parity check is said to be \emph{unsatisfied}; otherwise, it is satisfied.

We explore the structure of a local view of a generator in detail and establish notation as it is central for the decoding algorithm $\ssf$ (discussed in the next section).

\paragraph{The local view:}
Let $F \subseteq \supp(c_1,v_2)$ be a set of qubits.
Let $F_V$ be the restriction of $F \cap \Gamma(c_1) \times \{v_2\}$ on to the first component and $F_C$ be the restriction of $F \cap \{c_1\} \times \Gamma(v_2)$ on to the second component.
In other words, these are the restrictions of $F$ to the VV and CC qubits respectively.
We include a figure below to illustrate for a code constructed with a classical graph $G$ such that $\Delta_V = 3$ and $\Delta_C = 4$.
The sets $F_V$ and $F_C$ are denoted in green and blue respectively.
We let $F_V^c = \Gamma(c_1) \times \{v_2\} \setminus F_V$ and $F_C^c = \{c_1\} \times \Gamma(v_2) \setminus F_C$ denote their complements within the support of $(c_1,v_2)$.

\begin{center}
    \begin{tikzpicture}
    
    \fill[fill=green!20] (0,1.5) rectangle (0.5,2);
    \fill[fill=green!20] (0,2) rectangle (0.5,2.5);
    \fill[fill=blue!20] (1,0) rectangle (1.5,0.5);
    \fill[red!20] (1.5,1.5) rectangle (2,2);
    \fill[red!20] (1.5,2) rectangle (2,2.5);
    \fill[red!20] (1,1.5) rectangle (1.5,2);
    \fill[red!20] (1,2) rectangle (1.5,2.5);
    \fill[blue!20] (1,1) rectangle (1.5,1.5);
    \fill[blue!20] (1,2.5) rectangle (1.5,3);
    \fill[green!20] (2,1.5) rectangle (2.5,2);
    \fill[green!20] (2,2) rectangle (2.5,2.5);    
    \foreach \x in {1,1.5,2} {
    \foreach \y in {1,1.5,2,2.5} {
        \draw (\x,\y) rectangle ({\x+0.5},{\y+0.5});
    }
    }
    \foreach \x in {1,1.5,2} {
        \draw (\x,0) rectangle ({\x+0.5},0.5);
    }
    \foreach \y in {1,1.5,2,2.5} {
        \draw (0,\y) rectangle (0.5,{\y+0.5});
    }
    \draw (0.5,0.5) rectangle (0,0);
    
    \node[rotate=-45] at (-0.2,-0.2) {$(c_1,v_2)$};
    \node at (1.75,-0.5) {$\{c_1\} \times \Gamma(v_2)$};
    \node[rotate=90] at (-0.5,2) {$\Gamma(c_1) \times \{v_2\}$};
\end{tikzpicture}
\end{center}
The \emph{neighborhood} of a set $F \subseteq \supp(c_1,v_2)$, denoted $\Lambda(F)$, is the set of parity checks that is incident to at least one element of $F$.
The neighborhood can be expressed as the union $\Lambda(F) = F_V \times \Gamma(v_2) \union \Gamma(c_1) \times F_C$.
This corresponds to the union of the colored parity checks in the illustration above.

The \emph{unique} neighborhood of $F$, denoted $\Lambda^{(u)}(F)$, is the set of parity checks that is incident to one and only one element of $F$.
The unique neighborhood can be expressed as the disjoint union $\Lambda^{(u)}(F) = F_V \times F_C^C \sqcup F_V^c \times F_C$.
In the figure above, the unique neighborhood of $F$ is depicted in color corresponding to their unique neighbor---the neighbors of $F_V$ are in green and the neighbors of $F_C$ are in blue.
The parity checks in red are also called the multi-neighborhood of $F$; we do not introduce separate notation to refer to this set.

\subsection{Decoding algorithms}

A \emph{decoding algorithm} is a classical algorithm, which, given a syndrome bitstring $\synd$ as input, produces a correction consistent with $\synd$.
This corresponds to a bitstring $\ccorr \in \bbF_2^n$ and a Pauli operator $\widehat E$ for classical and quantum codes respectively.
We refer to these as corrections.

In the quantum case, given an error $E$, the decoder is passed the syndrome $\synd(E)$ and either \emph{halts} (it produces no output and declares failure), or outputs a correction $\widehat E$ with the same syndrome, i.e., $\synd(\widehat E) = \synd(E)$.
When the decoder produces an output, it succeeds if the net Pauli applied by both the error and the correction $Z(E) Z(\widehat E) \in \{ \pm i, \pm 1\} \cdot \cS$, and fails otherwise.

\noindent To recap notation we use throughout this paper:
\begin{itemize}
    \item $\cerr$, $\ccorr$ refers to classical errors interpreted as binary strings of length $n$.
    \item $E$, $\Ehat$ refers to subsets of $[N]$ corresponding to supports of Pauli operators.
    \item $\synd$ refers to the syndrome of both classical and quantum errors, e.g.\ $\synd(\cerr)$ or $\synd(E)$.
    \item $\oplus$ denotes the sum of vectors over $\bbF_2$.
    \item $\symmdiff$ denotes the symmetric difference of sets.
\end{itemize}

\subsubsection{Small Set Flip}
\label{sec:ssf}
Before describing $\ssf$ for quantum codes it is informative to briefly review the $\flip$ algorithm for classical codes which was introduced by Sipser and Spielman \cite{sipser1994expander}.
We consider a classical code $\cC$, defined in terms of a $(\Delta_V, \Delta_C)$-biregular bipartite Tanner graph $G = (V \union C,E)$.
Let $\cerr \in \bbF_2^n$ be an error and let $\synd$ be the corresponding syndrome.
The algorithm is only provided the syndrome $\synd$ and is initialized with $\ccorr = \mathbf{0}$.
$\flip$ is an iterative algorithm such that in each iteration, $\flip$ tries to find a vertex $v \in V$ to flip, i.e.\ it maps $\ccorr_v \mapsto \ccorr_v \oplus 1$, such that the syndrome weight decreases.
The algorithm terminates when no such vertex $v$ exists.
For the $\flip$ algorithm to succeed, when the algorithm terminates $\cerr = \ccorr$.
There are two modes of failure.
First, it can find a valid correction $\ccorr$ such that $\ccorr + \cerr$ is a logical error, which we call an error correction failure.
Second, it can terminate with a non-zero syndrome, a failure mode that we call a \emph{halt failure}.

To guarantee $\flip$ succeeds, we require that the bipartite graph $G$ that defines the code $\cC$ corresponds to an expander, i.e.\ for all $S \subseteq V$,
\begin{align}
     |S| \leq \alpha n \implies |\Gamma(S)| \geq (1-\epsilon)\Delta_V|S|~,
\end{align}
where $0 < \alpha < 1$ and $0 < \epsilon < 1/4$ are parameters of $G$.
The algorithm can correct errors $\cerr$ whose weight is a constant fraction of the code distance.

Small Set Flip ($\ssf$) is an iterative decoding algorithm for quantum codes that generalizes $\flip$.
In the iteration indexed by $i$, it applies an update to the correction and correspondingly updates the syndrome to $\synd^{(i)}$.
When it can no longer update the syndrome, it terminates.

For any $F \subseteq [N]$, we define the $\genscore$ function as
\begin{align}
\label{eq:score}
    \genscore(F) = \left\{\frac{|\synd^{(i)}| - |\synd^{(i)} \oplus \synd(F)|}{|F|} \right\}~,
\end{align}
where $\oplus$ denotes the sum mod $2$.
In the special case that $F$ is empty, we define $\genscore(F) = 0$.

The score of a set $F$ captures how good a partial $Z$ type correction $F$ is.
The numerator of the score function corresponds to the change in the syndrome weight that would be caused by flipping $F$.
The denominator, on the other hand, assigns a higher score to $F$ if its weight is lower.
In particular, if we have two candidates $F$ and $F'$ that cause the same change in syndrome but obey $|F| < |F'|$, then $\score(F) > \score(F')$.

Let $\cF$ denote the union of all subsets $F$ contained in the support of any $Z$ type generator:
\begin{align}
    \cF = \{F : F \subseteq \supp(S), S \in \cZ \}~.
\end{align}

$\ssf$, as specified in Algorithm~\ref{alg:ssflip}, is an iterative algorithm where in each iteration, the \emph{score} is evaluated for each element $F \in \cF$, and then a set with the maximum score is applied as a partial correction (if positive), which reduces the syndrome weight.

\begin{algorithm}[h]
	\begin{algorithmic}[0]
		\State \textbf{Input:} A syndrome $\synd(E) \in \bbF_2^{M}$.
		\State \textbf{Output:} Correction $\Ehat$
    \end{algorithmic}
    \begin{algorithmic}[1]\onehalfspacing
        \State $i \leftarrow 0$
        \State $\synd^{(0)} \leftarrow \synd(E)$
        \State $\Ehat \leftarrow \emptyset$
        \State $F_0 \leftarrow \argmax \genscore(F)$ over $F \in \cF$
		\While{$\genscore(F_i) > 0$}
            \State $\Ehat \leftarrow \Ehat \symmdiff F_i $
            \State $\synd^{(i+1)} \leftarrow \synd^{(i)} \oplus \synd(F_i)$
            \State $i \leftarrow i+1$
            \State $F_i \leftarrow \argmax \genscore(F)$ over $F \in \cF$
		\EndWhile
		\State \Return $\Ehat$.
	\end{algorithmic}
	\caption{$\ssf(\synd(E)$) (Informal)}
	\label{alg:ssflip}
\end{algorithm}

It may appear at first that this algorithm formally has run time $O(N \log N)$  rather than $O(N)$ since the list of scores needs to be stored and updated in order to pick out the maximum at each round, and comparison sorting requires time $\Omega(N \log N)$.
However, there is a finite set of allowed scores since the size of each element in $\cF$ is limited by the size of the largest stabilizer generator which also implies a finite set of allowed syndrome differences in each iteration.
Sorting a set of $N$ elements which only take values in a finite range only requires time $O(N)$ using \ttt{counting sort}~\cite{cormen01}.
In practice, this makes little difference in the run time of instances in the regimes we consider in this paper and we use a standard comparison-based sorting algorithm in our implementations.
Note that we provide a more formal version of Algorithm~\ref{alg:ssflip}, namely Algorithm~\ref{alg:ssflip-detailed} in \app{space-time-cost}, which includes further details of the data structures used.

While $\ssf$ is a linear-time algorithm, the dependence of the time complexity on the degree of the quantum LDPC code scales poorly.
For a specific generator $S \in \cZ$, evaluating the scores takes $\Theta(2^{|\supp(S)|})$ time since that generator contributes $2^{|\supp(S)|}$ elements to $\cF$.
For a fixed family of LDPC codes, this quantity is a constant as $\Delta_V$ and $\Delta_C$ are constants, but it can be large.

There are two modes of failure for the $\ssf$ decoding algorithm.
First, it can find a valid correction $\Ehat$ such that $Z(E)\cdot Z(\Ehat)$ is a logical error, which we call an error correction failure.
Second, it can terminate with a non-zero syndrome, a failure mode that we call a halt failure.
We count both as logical failures.
For a fixed code $\cQ$ in the code family, an error $E$ on which $\ssf$ cannot complete even one iteration is called a stopping set.
Stopping sets have been explored in the context of other iterative decoding algorithms and other quantum codes (e.g.\ \cite{raveendran2021trapping}).
When $\ssf$ terminates with a halt failure, the stopping sets correspond to stopping sets of $\ssf$; we refer to the corresponding syndromes as stopping syndromes.

\subsubsection{Belief Propagation and Ordered Statistics Decoding}

Here we provide high-level overviews of a Belief Propagation ($\bp$) decoder and an Ordered Statistics Decoder ($\osd$) for classical codes.
We use these decoders as subroutines for the $\pal$ decoder which we introduce in \sec{pal}.
For simulations in this paper, we use the open-source implementation of $\bp$ and $\osd$ in the software package by Joschka Roffe \cite{roffe_decoding_2020} with specific parameter settings as described in \app{details}.

The term `$\bp$ decoder' is really an umbrella term for a class of decoding algorithms on Tanner graphs, which calculate the marginal probability of an error on each bit given the observed syndrome.
Since the marginal probability of an error on a given bit depends on the probability of errors on other bits, information is iteratively passed between neighbors in the graph, to approximate the most probable distribution over bits.
This process ideally converges to a consistent probability distribution that can be used to infer a correction.
In this paper, we use the implementation from \cite{roffe_decoding_2020} which is a min-sum implementation of $\bp$, and henceforth refer to that implementation simply as $\bp$.

If the physical failure probability of each bit is decreased, we expect the logical failure probability to decrease as well.
However, $\bp$ can instead run into an `error floor'---where the logical failure rate plateaus rather than decreasing with further iterations.
This can be caused by feedback loops in $\bp$ due to the existence of cycles in the Tanner graph~\cite{poulin2008iterative}.

The Ordered Statistics Decoder~\cite{fossorier1995soft} is a classical decoder by Fossosier and Lin, which when generalized to the quantum setting can overcome the problem of feedback loops that are experienced by $\bp$~\cite{panteleev2021degenerate} and can perform well in practice for quantum codes~\cite{roffe2020decoding}.
Here we focus on $\osd$ as a classical decoder however.

The $\osd$ algorithm involves separating the $n$ bits of the code into two sets: a set $T$ of $n-m$ bits which we are confident have \emph{not} been flipped, and a set $S$ of $m$ bits which may or may not have been flipped. 
The set $S$ the code's $n \times m$ parity check matrix $\pcm$ restricted to columns in $S$ is invertible (i.e. the $m$ columns are linearly independent).
First, we run $\bp$ and the columns of $\pcm$ are sorted (say, from left to right) according to the confidence with which we believe they are flipped.
This results in a partition of $\pcm$ (up to column permutations) of the following form
\begin{align*}
    \pcm = (\pcm_S | \pcm_T)~,
\end{align*}
where the columns of $\pcm_S$ are linearly independent, i.e.\ $\pcm_S$ is invertible and all the bits in $S$ correspond to bits that are flipped.
In this ordering, we more strongly believe that the bit indexed by the $j$\textsuperscript{th} column of $\pcm$ has been flipped than the bit indexed by the $(j+1)$\textsuperscript{th} column of $\pcm$.
A solution $\ccorr$ to the equation $\pcm \cerr = \synd$ can be expressed as $\ccorr = (\ccorr_S|\ccorr_T)$.
As $\pcm_S$ is invertible, for every choice of $\ccorr_T$, there exists a unique choice of $\ccorr_S$ such that $\pcm \ccorr = \synd$.

As the columns of $\pcm_S$ are linearly independent, we can find a solution of the form $\ccorr^{(0)} = (\ccorr_S^{(0)}| \mathbf{0})$.
However, this solution need not be the most likely solution---there may exist a correction $\ccorr' = (\ccorr_S'| \ccorr_T')$ with lower weight than $\ccorr^{(0)}$.
To find the most likely solution, one iterates through every possible choice of $\ccorr_T$.
For a constant rate code, this algorithm requires exponential time as $\ccorr_T$ is supported on $n-m = \Theta(n)$ bits.
However, there exist efficient approximations that only use low weight strings $\ccorr_T$.
We use the variant $\bposd$-$\mathtt{CS}$ by Roffe \emph{et al}.\ \cite{roffe2020decoding} that we refer to simply as $\bposd$.
In this variant, all weight $2$ solutions in $T$ are considered.

This algorithm runs in $O(n^3)$ time for codes of size $n$, where the dominant cost is Gaussian elimination for the matrix inversion step.

\subsection{Numerical tests of decoder run time and error correction performance}

The HGP codes which we use for our examples are specified as follows.
Bipartite graphs with bounded degree corresponding to classical codes are constructed using the Progressive Edge Growth (PEG) algorithm \cite{hu2001progressive}.
The classical codes themselves correspond to $(3,4)$-biregular bipartite graphs; this results in quantum codes with $\Delta_q \leq 8$, $\Delta_g \leq 7$.

In this work we assume that $X$ and $Z$ errors are decoded independently, and since the codes are symmetric under the exchange of $X$ and $Z$ we consider only $Z$ errors.
We assume that each qubit independently undergoes a $Z$ error with probability $\pphys$, resulting in a $Z$ type error specified by a bitstring $E$, which has an $X$-type syndrome $\synd(E)$, which we assume is extracted perfectly.
We then run the implementation of the decoder on this syndrome, which either produces a proposed correction $\widehat E$, or experiences a halting failure.
If the combined effect of the proposed correction and the error, namely $Z(E) \cdot Z(\widehat E)$, is a $Z$ stabilizer, the decoder succeeds; otherwise, we say that a \emph{logical} failure has occurred.

We adhere to the following rules in order to take data for numerical comparisons:
\begin{enumerate}
    \item We present the average costs over all samples.
    \item We take sufficiently many samples such that the statistical uncertainty of failure probabilities appear only at the second significant figure.
\end{enumerate}

\section[Decoding after SSF halt failures]{Decoding after $\ssf$ halt failures}
\label{sec:pal}

\begin{table*}
    \centering
    \resizebox{\textwidth}{!}{%
    \begin{tabular}{|c||c|c|c||c|c|c||c|c|c|}
    \hline
    Error rate &  \multicolumn{3}{|c||}{$p = 1$\%}  & \multicolumn{3}{|c||}{$p = 2$\%} & \multicolumn{3}{|c|}{$p = 3$\%} \\
    \hline
    Code size & 900 & 1600 & 3600 & 900 & 1600 & 3600 &   900 & 1600 & 3600 \\
    \hline
    \hline
     halt fail & \multirow{2}{*}{0.999(5)} & \multirow{2}{*}{0.999(4)} & \multirow{2}{*}{0.999(2)} & \multirow{2}{*}{0.999(3)} & \multirow{2}{*}{0.999(4)} & \multirow{2}{*}{0.999(1)} & \multirow{2}{*}{0.999(3)} & \multirow{2}{*}{0.999(3)} & \multirow{2}{*}{0.999(1)}  \\
     fraction &&&&&&&&&  \\
    \hline
     \multirow{2}{*}{1 line} & \multirow{2}{*}{0.955(4)} & \multirow{2}{*}{0.96(4)} & \multirow{2}{*}{0.97(4)} & \multirow{2}{*}{0.81(4)} & \multirow{2}{*}{0.79(6)} & \multirow{2}{*}{0.79(6)} & \multirow{2}{*}{0.57(5)} & \multirow{2}{*}{0.50(5)} & \multirow{2}{*}{0.44(5)}  \\
      &&&&&&&&&  \\
    \hline
     \multirow{2}{*}{2 lines} & \multirow{2}{*}{0.016(8)} & \multirow{2}{*}{0.022(3)} & \multirow{2}{*}{0.02(4)} & \multirow{2}{*}{0.05(2)} & \multirow{2}{*}{0.09(4)} & \multirow{2}{*}{0.12(5)} & \multirow{2}{*}{0.10(3)} & \multirow{2}{*}{0.15(4)} & \multirow{2}{*}{0.21(5)}  \\
      &&&&&&&&&  \\
    \hline
     \multirow{2}{*}{$3+$ lines} & \multirow{2}{*}{0.027(3)} & \multirow{2}{*}{0.021(3)} & \multirow{2}{*}{0.01(3)} & \multirow{2}{*}{0.13(4)} & \multirow{2}{*}{0.12(5)} & \multirow{2}{*}{0.09(5)} & \multirow{2}{*}{0.32(5)} & \multirow{2}{*}{0.34(5)} & \multirow{2}{*}{0.35(5)}  \\
 &&&&&&&&&  \\
    \hline
    \end{tabular}%
    }
    \caption{
    For a range of error rates and HGP code instances, we show the fraction of logical failures which result from $\ssf$ halts.
    After each halt failure, we find the number of lines which support the stopping syndrome, and show here the fraction which fit on one line versus two or more lines.
    Numbers are reported using significant digits, where the number in parentheses indicates the leading digit which is not known with 95 \% confidence.
    }
\label{tab:line_fractions_2}
\end{table*}

The $\ssf$ decoder has proven performance guarantees for HGP codes with sufficiently large input graphs with sufficiently large degree \cite{leverrier2015quantum} \footnote{Formally, these guarantees are expressed in terms of a property called \emph{graph expansion}; $\ssf$ is guaranteed to converge for graphs with sufficient expansion.}.
However, for practical purposes we are more interested in HGP codes built from small input graphs with more modest degree.
In this regime $\ssf$ often does not find a correction and instead terminates with a stopping failure; see Table~\ref{tab:line_fractions_2}.

In the standard version of $\ssf$ presented in \sec{background}, stopping failures are treated as bona fide failures.
Before it stops, however, $\ssf$ can undergo many rounds resulting in a partial correction that significantly reduces the syndrome weight compared with that of the initial error.

In \sec{stopping-fail-structure} we show that the stopping syndromes tend to exhibit a specific structure.
In \sec{pal-decoder} we propose a new decoding algorithm that can often exploit this structure to correct the stopping set.
We call this the Projection-Along-a-Line ($\pal$) decoder.
In \sec{pal-numerics} we analyze the performance of the two-step decoder formed by applying $\ssf$ followed by $\pal$, finding order-of-magnitude improvements over $\ssf$ alone in regimes of interest.

\subsection{Understanding stopping failures}
\label{sec:stopping-fail-structure}

Here we seek to understand syndrome patterns that cause the $\ssf$ decoder to halt for HGP codes.

As explained in \sec{ssf}, $\ssf$ is an iterative algorithm which monotonically reduces the syndrome weight over a sequence of rounds by applying a partial correction in each round until either the syndrome is completely removed or a stopping syndrome is reached.
Stopping syndromes have the property that none of the potential corrections offered by $\ssf$ can reduce their syndrome weight.
As fixed points of the $\ssf$ algorithm, stopping syndromes are much more structured than syndromes corresponding to random errors applied by the noise model.
Furthermore, many different initial syndromes generated by random errors can flow to the same stopping syndrome under $\ssf$.

To understand the structure of stopping syndromes it is informative to consider the structure of the HGP codes that we reviewed in \sec{hgp}.
For concreteness, consider an HGP code associated with a graph $\cG = G_1 \times G_2$, formed from two copies ($G_1$ and $G_2$) of the bipartite graph $G = (V \union C,E)$, which we assume is $(\Delta_V,\Delta_C)$ biregular.
Let $\cC$ ($\tC$) be the code obtained by associating $V$ ($C$) with the bits and $C$ ($V$) with the parity checks.

Since we consider $Z$-type errors (which can occur on either CC or VV qubits), these are only detected by VC checks.
The product structure of $\cG$ implies that for any $v_1 \in V_1$ or $c_2 \in C_2$, each \emph{horizontal subgraph} $\{v_1\} \times G_2$ and each \emph{vertical subgraph} $G_1 \times \{c_2\}$ is isomorphic to $G$.
We say that any syndrome entirely supported on $\{v_1\} \times C_2$ vertices of a horizontal subgraph is a \emph{horizontal line}.
Similarly, we say that a syndrome supported on $V_1 \times \{c_2\}$ vertices of a vertical subgraph form a \emph{vertical line}.
Together, vertical and horizontal lines form \emph{line-like} syndromes which are very different from syndromes created by random errors.

\begin{figure}[h]
    \centering
    \begin{tikzpicture}
    \begin{scope}[scale=0.75]
        \draw (0,0) rectangle (3,3);
        \node at (-1.5,1.5) {$C_1$};
        \draw (0,3.2) rectangle (3,6.2);
        \node at (-1.5,4.7) {$V_1$};
        \draw (3.2,0) rectangle (6.2,3);
        \node at (1.5,7.5) {$V_2$};
        \draw (3.2,3.2) rectangle (6.2,6.2);
        \node at (4.7,7.5) {$C_2$};
        
        \foreach \x in {3.2,3.7,4.2,4.7,5.2,5.7} {
        \foreach \y in {3.2,3.7,4.2,4.7,5.2,5.7} {
            \draw ({\x},\y) rectangle ({\x+0.5},{\y+0.5});
         }
        }
        \foreach \x in {0,0.5,1,1.5,2,2.5} {
        \foreach \y in {3.2,3.7,4.2,4.7,5.2,5.7} {
            \draw ({\x},\y) rectangle ({\x+0.5},{\y+0.5});
         }
        }
        \foreach \x in {3.2,3.7,4.2,4.7,5.2,5.7} {
        \foreach \y in {0,0.5,1,1.5,2,2.5} {
            \draw ({\x},\y) rectangle ({\x+0.5},{\y+0.5});
         }
        }
        \foreach \x in {0,0.5,1,1.5,2,2.5} {
        \foreach \y in {0,0.5,1,1.5,2,2.5} {
            \draw ({\x},\y) rectangle ({\x+0.5},{\y+0.5});
         }
        }

        \foreach \y in {0,0.5,1,1.5,2,2.5} {
            \draw (-1,\y) rectangle (-0.5,{\y+0.5});
         }

         \foreach \y in {3.2,3.7,4.2,4.7,5.2,5.7} {
            \draw (-1,\y) rectangle (-0.5,{\y+0.5});
         }

         \foreach \x in {0,0.5,1,1.5,2,2.5} {
            \draw (\x,6.5) rectangle ({\x+0.5},7);
         }

         \foreach \x in {3.2,3.7,4.2,4.7,5.2,5.7} {
            \draw (\x,6.5) rectangle ({\x+0.5},7);
         }

         \fill[pattern=north west lines] (-1,3.7) rectangle (-0.5,4.2);
         \node at (-1.35,3.95) {$v_1$};

        \foreach \x in {0,0.5,1,1.5,2,2.5} {
            \fill[pattern=north west lines] (\x,3.7) rectangle ({\x+0.5},4.2);
         }

         \foreach \x in {3.2,3.7,4.2,4.7,5.2,5.7} {
            \fill[pattern=north west lines] (\x,3.7) rectangle ({\x+0.5},4.2);
         }
          \foreach \y in {0,0.5,1,1.5,2,2.5} {
            \fill[pattern=north west lines] (3.7,\y) rectangle (4.2,{\y+0.5});
         }

         \foreach \y in {3.2,3.7,4.2,4.7,5.2,5.7,6.5} {
            \fill[pattern=north west lines] (3.7,\y) rectangle (4.2,{\y+0.5});
         }
         \node at (3.95,7.35) {$c_2$};
         \end{scope}
    \end{tikzpicture}
    \caption{Schematic for vertical line corresponding to $c_2 \in C_2$ and horizontal line corresponding to $v_1 \in V_1$ in the HGP code.}
    \label{fig:exp-code2}
\end{figure}
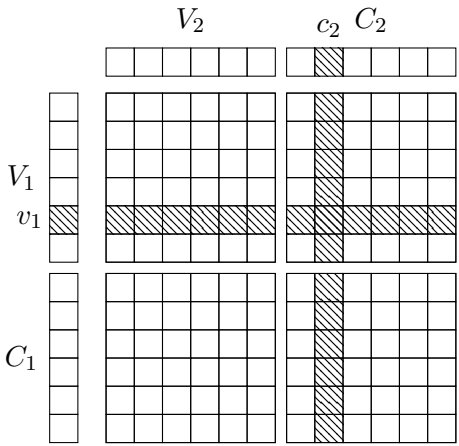

To gain some intuition for the structure of stopping syndromes, let us consider a simple scenario in which the initial syndrome which is provided to the $\ssf$ decoder is line-like.
To be precise, let $v_1 \in V_1$ and $\{v_1\} \times G_2$ identify the horizontal subgraph such that the syndrome is supported on $\{v_1\} \times C_2$.
Albeit contrived, it allows us to consider all the different cases that $\ssf$ considers in detail and illustrates how the decoder works.
In \sec{lemma-proof} we prove the following lemma.

\begin{lemma}
\label{lem:line-like}
    If the syndrome is line-like such that it is restricted to $\{v_1\} \times C_2$ for a fixed $v_1 \in V_1$, then the only valid sets $F$ that $\ssf$ can flip are supported on $\{v_1\} \times V_2$.
    Similarly, if it is line-like such that 
    it is restricted to $C_1 \times \{v_2\}$ for a fixed $v_2 \in V_2$, then the only valid sets $F$ that $\ssf$ can flip are supported on $C_1 \times \{v_2\}$.
\end{lemma}

In other words, \emph{if} there is a set $F$ that can reduce the weight of the syndrome, then this set $F$ is supported on $\{v_1\} \times V_2$.
Any $Z$ stabilizer generator can overlap at most one qubit from the set $\{v_1\} \times V_2$.
Therefore, this implies that the $\ssf$ decoder proceeds by iteratively applying single-qubit corrections within the line to reduce the syndrome weight --- hence the action of $\ssf$ on a line-like syndrome is equivalent to the $\flip$ decoder for the classical code $\cC$.
Therefore all initial syndromes that are line-like will either be removed by $\ssf$ or will flow to line-like stopping syndromes, which correspond to stopping syndromes for the $\flip$ decoder of the underlying classical code. 
This makes it clear that one class of stopping syndromes for $\ssf$ are line-like (and are stopping syndromes for the $\flip$ decoder for the associated classical code).
While we arrived at this class of stopping sets from the contrived starting point of a line-like initial syndrome, our numerical studies show that in fact a very large fraction of stopping syndromes are of this type given stochastic initial errors (which rarely have line-like initial syndromes).
This can be seen in Table~\ref{tab:line_fractions_2} which shows the frequency of the minimal number of lines needed to cover the support of the stopping syndromes observed for $\ssf$ for a variety of HGP code sizes and error rates.

We end this subsection with some (non-rigorous) intuition for why one might expect stopping syndromes to be line-like for unstructured random errors when the error rates are sufficiently low.
When error rates are low, the errors are likely to be isolated and far from each other with respect to the graph $\cG$, which will produce a syndrome which forms small clusters, with clusters being far apart with respect to the graph $\cG$.
Clusters of syndromes that are far apart in the graph $\cG$ are handled independently by $\ssf$, which only applies corrections within the support of generators, and these clusters will tend to shrink as the algorithm proceeds, many of them disappearing.
At some point, the syndrome will be very sparse, and will most likely lie on a small number of lines just by virtue of having low density, most of these line-like syndromes can be expected to be uncoupled as far as $\ssf$ is concerned because they are far apart in the graph $\cG$.
If the syndrome restricted to a line is uncoupled from the rest of the system, then it will be treated by $\ssf$ as if it were being handled by $\flip$ on the classical code as discussed above.
Eventually only those clusters which are line-like and stopping syndromes for the classical codes remain. 
As we expect these classical stopping syndromes to be rare, probably none will arise and $\ssf$ will simply finish by producing a valid correction.
With some small chance, one or more classical stopping syndromes will have been produced (causing a stop failure for $\ssf$). 
If the qubit failure rate is increased beyond this low error rate regime, errors form larger clusters that are not very well separated on the graph.
The stopping sets begin to look more complex as errors on multiple lines begin to interact.
The structure of the stopping syndromes ceases to have an easily discernible form.

\subsection[PAL decoding algorithm]{$\pal$ decoding algorithm}
\label{sec:pal-decoder}

Our algorithm proceeds iteratively.
First, we express the syndrome as contributions from different lines and these are sorted into $\hor$ and $\vert$ for horizontal and vertical lines respectively.
Points are counted as both horizontal and vertical lines.
For each line $\ell \in \hor \union \vert$, we obtain a correction $F_{\ell}$.
We then apply the correction with the maximum score, i.e.\ the correction $F_{\ell}$ that maximizes the score defined in Eq.~\ref{eq:score}
\begin{align*}
    \score(F_{\ell}) = \frac{|\synd| - |\synd \oplus \synd(F_{\ell})|}{|F_{\ell}|}~.
\end{align*}

Algorithm~\ref{alg:get_lines} presented below describes how we classify syndromes into lines.
\begin{algorithm}[ht]
    \begin{algorithmic}[0]
        \State \textbf{Input:} Syndrome $\synd \in \bbF_2^{V_1 \times C_2}$, 
        \State \textbf{Output:} $\hor$, $\vert$.

    \end{algorithmic}
    \begin{algorithmic}[1]
        \State $\hor, \vert \leftarrow \emptyset$ will be the set of horizontal and vertical lines.
        \For{$v \in V$}
            \State $\hor \leftarrow \hor \union (\{v\} \times C \cap \synd)$
        \EndFor
        \For{$c \in C$}
            \State $\vert \leftarrow \vert \union (V \times \{c\} \cap \synd)$
        \EndFor
        \State\Return $\hor, \vert$
    \end{algorithmic}
    \caption{$\ttt{GetLines}(\synd)$}
    \label{alg:get_lines}
\end{algorithm}
In our implementation of Algorithm~\ref{alg:get_lines}, we use a dictionary to decrease the run time from $O(|\synd|^2)$, as specified in the pseudocode, to $O(|\synd|)$.

\ttt{GetLines} accepts as input the syndrome $\synd \in \bbF_2^{V_1 \times C_2}$.
For each $v_1 \in V_1$, it scans through $\synd$, and if it finds any unsatisfied parity checks supported on the subset $\{v_1\} \times C_2$, then it classifies this as a \emph{horizontal} line.
It repeats a similar process for every $c_2 \in C_2$ and scans for unsatisfied parity checks supported on $V_1 \times \{c_2\}$ \footnote{
We find $\hor$ and $\vert$ may contain elements in common corresponding to syndromes that are isolated points.}.

\begin{algorithm}[h]
    \begin{algorithmic}[0]
        \State \textbf{Input:} Syndrome $\synd \in \bbF_2^{V_1 \times C_2}$, 
        \State \textbf{Output:} Correction $\Ehat$
    \end{algorithmic}
    \begin{algorithmic}[1]\onehalfspacing
        \State $i \leftarrow 0$
        \State $\Ehat \leftarrow \emptyset$.
        \State $\hor, \vert \leftarrow \ttt{GetLines}(\synd)$
        \State $\lines \leftarrow |\hor| + |\vert|$.
       \While{$\lines > 0$ and $i < \scrI$}
        \State $i \leftarrow i+1$
        \State $F_i \leftarrow 0$
        \For{$\synd_j \in \hor \union \vert$}
            \State $F_{j} \leftarrow$ $\bposd(\synd_j)$ \label{line:bposd}.
            \If{$\score(F_j) > \score(F_i)$}
                \State $F_i \leftarrow F_j$
            \EndIf
        \EndFor
        \State $\Ehat \leftarrow \Ehat \symmdiff F_i$
        \State $\synd \leftarrow \synd \oplus \synd(F_i)$
        \State $\hor, \vert \leftarrow \ttt{GetLines}(\synd)$
        \State $\lines \leftarrow |\hor| + |\vert|$.
        \EndWhile
        \State \Return $\Ehat$.
    \end{algorithmic}
    \caption{$\pal(\synd)$}
    \label{alg:pal}
\end{algorithm}

The $\pal$ algorithm, described in Algorithm~\ref{alg:pal}, is an iterative decoding algorithm that takes the classification into lines as input.
It decodes each set in $\hor \union \vert$ separately using $\bposd$ as stated in Line \ref{line:bposd}.
These lines are interpreted as syndromes corresponding to errors on classical codes; an element of $\hor$ of the form $\{v_1\} \times C_2$ will only return a correction within $\{v_1\} \times V_2$.
Similarly, an element in $\vert$ of the form $V_1 \times \{c_2\}$ will only return a correction within $C_1 \times \{c_2\}$.
To be precise, we define two versions of $\bposd$ in our implementation---one used for $\hor$ and the other used for $\vert$.
Line \ref{line:bposd} does not include this complexity for conceptual clarity and ease of presentation.
We highlight that any decoding algorithm for classical codes can be used here in principle.
$\bposd$ has the best performance among the heuristics we tried.

In each iteration, we flip the line whose correction has the largest score.
This continues until the syndrome is the $0$ string or until the total number of iterations is $\scrI$.
As the physical error rate increases, the number of required iterations grows.
In our simulations, we set $\scrI = 20$ as an upper bound.

\noindent The complete decoding algorithm is presented below in Algorithm~\ref{alg:decode-update}.
\begin{algorithm}[h]
    \begin{algorithmic}[0]
        \State \textbf{Input:} Syndrome $\synd \in \bbF_2^{V_1 \times C_2}$
        \State \textbf{Output:} Correction $\Ehat$ if algorithm converges and FAIL otherwise.
    \end{algorithmic}
    \begin{algorithmic}[1]\onehalfspacing
        \State $\Ehat_{\ssf} \leftarrow \ttt{SSF}(\synd)$
        \State $\Ehat_{\pal} \leftarrow \pal(\synd(\Ehat_{\ssf}))$
        \If{$\synd(\Ehat_{\ssf} \symmdiff \Ehat_{\pal}) \oplus \synd$ is zero}
            \State \Return $\Ehat \leftarrow \Ehat_{\ssf} \symmdiff \Ehat_{\pal}$
        \Else
            \State \Return FAIL (halt).
        \EndIf
    \end{algorithmic}
    \caption{\ttt{Decode}($\synd$)}
    \label{alg:decode-update}
\end{algorithm}

To run $\bposd$ on each classical code, we need to know how to initialize log-likelihood ratios (LLRs) for each line.
Setting the initial LLRs for $\bposd$ is not trivial.
Indeed, for a fixed qubit failure rate $\pphys$, $\ssf$ may manage to correct at least some of the initial errors before it terminates.
Therefore, the distribution of errors along a line can be very different from the physical failure rate $\pphys$.
For example, errors may be ``concentrated'' along a specific line, i.e.\ while the total number of errors in the HGP code decreases, some lines may end up with more errors than they began with.
In \app{details}, we include the details required to initialize and run $\bposd$ along each line.

One could consider a further optimized version where we could infer the physical bit error rates $q_V$ and $q_C$ as a function of both $\pphys$ but also as a function of the number of syndromes along that line.
Instead, we take an average over all possible lines.

\paragraph{Time complexity of $\ssf+\pal$:}
The time complexity of the $\ssf$ decoder is $O(N)$.
The complexity of the $\pal$ decoder depends on the number of lines and on the complexity of $\bposd$.
The latter is straightforward -- it is $O(N^{3/2})$ per line as $\bposd$ runs in $O(n^3)$ time for codes of block length $n$ and each line is a code of length $O(\sqrt{N})$.
However, as the number of lines is itself a random variable, we can only heuristically estimate the complexity of $\ssf+\pal$.
Algorithm~\ref{alg:get_lines} takes time $O(N)$.
The number of lines after each round of $\pal$ does not increase. 
To see why, note that applying any correction along a line does not increase the number of lines.

With this observation, suppose there are $L$ lines.
Recall that $\pal$ is permitted to run for a maximum number of iterations $\scrI$.
Therefore, the total complexity of $\bposd$ is at most $O(\scrI \cdot L \cdot N^{3/2})$.

\subsection{Results of numerical simulations}
\label{sec:pal-numerics}

The results of our simulations are presented in \fig{ssf-pal-34}.
For low qubit failure rates, $\ttt{PAL}$ improves the decoding performance by at least an order-of-magnitude.
For instance, consider the HGP code obtained as a product of $(3,4)$ biregular graphs.
At $1$\% physical error rates, the logical error probability for the code with $3600$ physical qubits falls from roughly $5 \times 10^{-2}$ to $4 \times 10^{-4}$.
The multiple by which the logical error rate decreases improves as we decrease the qubit failure rate.

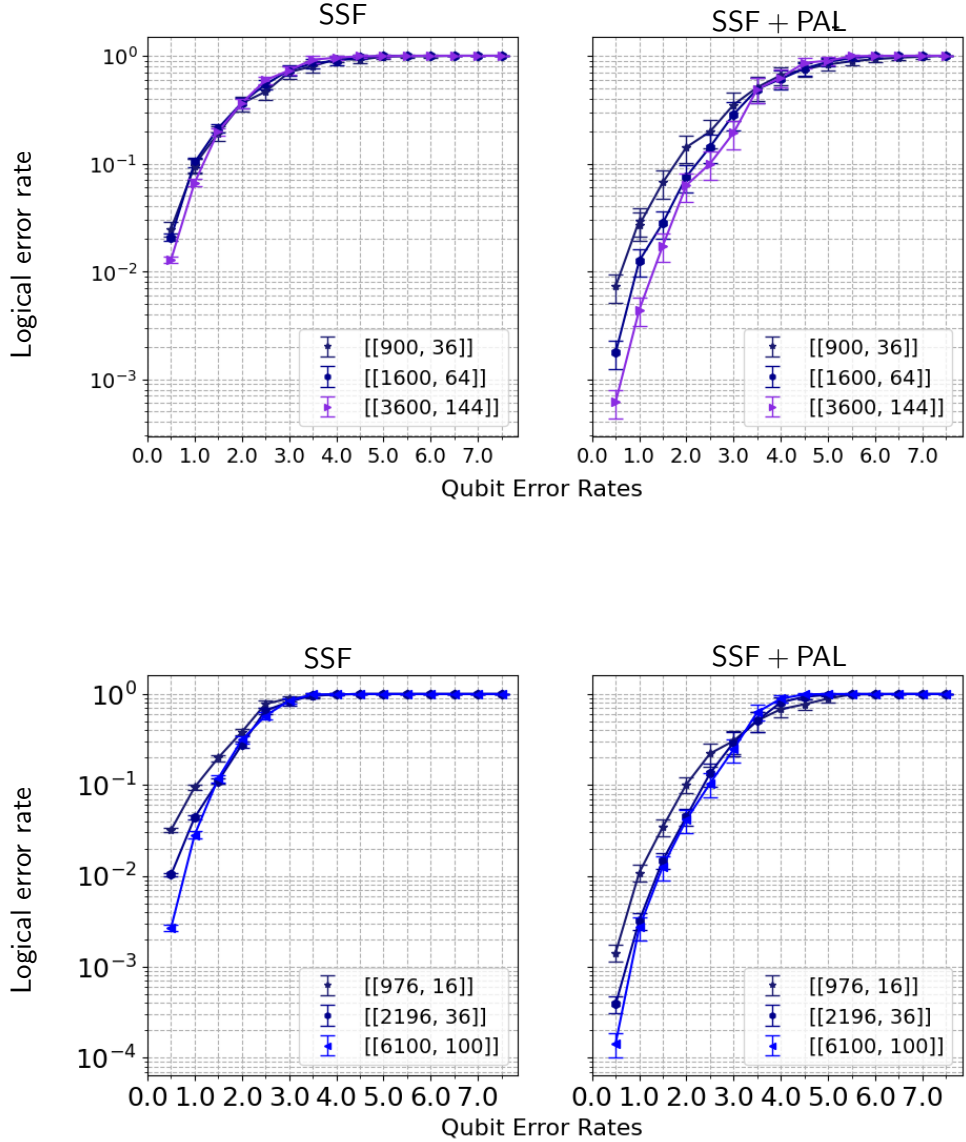
\begin{figure*}
    \centering
    \begin{tikzpicture}
        \node at (0,0) {\includegraphics[scale=0.55]{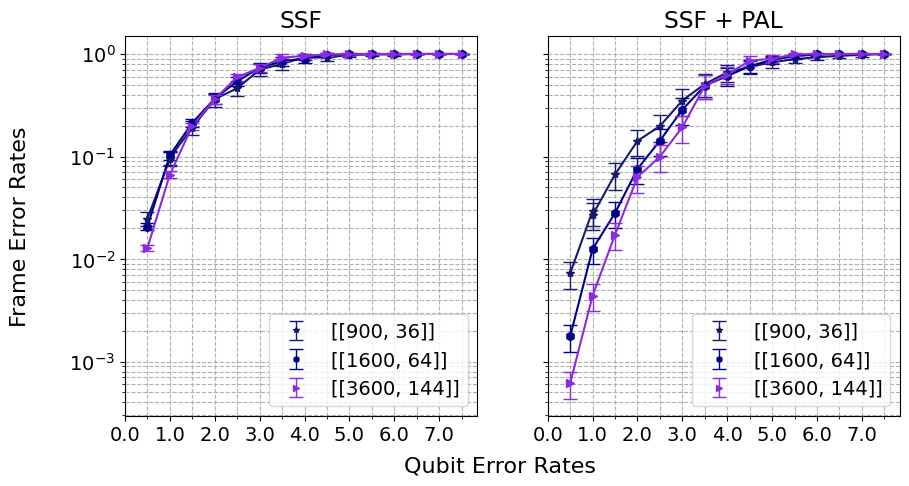}};
        \fill[white] (-2.65,3.5) rectangle (-1.5,3);
        \node at (-2,3.2) {$\mathsf{SSF}$};

        \fill[white] (2.8,3.6) rectangle (4.5,3);
        \node at (3.8,3.1) {$\mathsf{SSF + PAL}$};
        
        \fill[white] (-6.5,-1.5) rectangle (-5.75,2);
        \node[rotate=90] at (-6.25,0) {$\mathsf{Logical}$ $\mathsf{error}$ $\mathsf{rate}$};
    \end{tikzpicture}
    \begin{tikzpicture}
        \node at (0,0) {\includegraphics[scale=0.55]{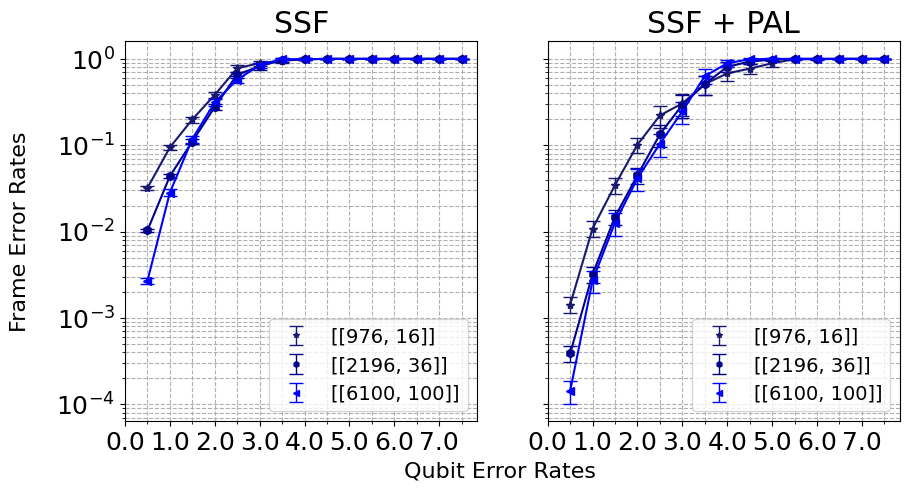}};
        \fill[white] (-6.5,-1.5) rectangle (-5.75,2);

        \fill[white] (0,4.5) rectangle (5.5,2.95);
        \node at (3.8,3.1) {$\mathsf{SSF + PAL}$};
        
        \fill[white] (-3,4.9) rectangle (-0.5,2.95);
        \node at (-2.2,3.1) {$\mathsf{SSF}$};

        \node[rotate=90] at (-6.25,0) {$\mathsf{Logical}$ $\mathsf{error}$ $\mathsf{rate}$};
    \end{tikzpicture}
    \caption{
    Comparing the change in logical error probability before and after the post-processing algorithm $\ttt{PAL}$ for the hypergraph product code constructed from $(3, 4)$-biregular bipartite graphs (upper panels) and $(5,6)$-biregular bipartite graphs (lower panels).
    This results in a family of quantum codes with $\Delta_q = 8,\Delta_g = 7$ (upper panels) and $\Delta_q = 12,\Delta_g=11$ (lower panels).
    The error bars represent the 95\% confidence intervals, i.e.\ $\approx 1.96$ standard deviations.
    We note that there is not a very smooth improvement of the curves for successive codes for the $(5,6)$ case.
    We believe this is due to the random procedure used to generate the classical codes.
    } 
    \label{fig:ssf-pal-34}
\end{figure*}

\begin{figure*}[h]
    \centering

    \label{fig:ssf-pal-56}
\end{figure*}

\section{Conclusions and future work}

We have studied failures for Small-Set-Flip ($\ssf$) on hypergraph product codes.
We focused on stopping failures, instances where $\ssf$ terminates without entirely correcting the error.
The error when $\ssf$ terminates was called the stopping set and the corresponding syndrome the stopping syndrome.
We found that the stopping syndromes supported on lines in the hypergraph product code.
A line corresponds to a subgraph of the hypergraph product code of the form $\{v_1\} \times C_2$ for $v_1 \in V_1$ or $V_1 \times \{c_2\}$ for $c_2 \in C_2$.
As these lines correspond to classical codes $\cC$ and $\tC$, this allowed us to reduce the problem of decoding a quantum code to decoding a classical code.
To this end, we propose using the Projection-Along-a-Line decoder ($\pal$) after $\ssf$.

The algorithm $\pal$ is itself an iterative algorithm; in each iteration, it identifies corrections along a line.
As the stopping syndrome can be supported on lines that share parity checks, the corrections along different lines are not independent.
We use an appropriately defined score function to pick one line and the corresponding correction.
Intuitively, this corresponded to the correction that reduces the syndrome the most with least number of flips.

While $\pal$ allows for any classical decoding algorithm to be used as a subroutine, we have used $\bposd$.
We found that for error rates of $1\%$ or lower, the logical failure rate can improve by an order-of-magnitude.

One can consider the use of $\pal$ as a post-decoder for other decoders.
In \app{bp-predec} we compare the performance of the $\bp+\ssf$ decoder with $\bp+\ssf + \pal$ with HGP codes built from (3,4)-biregular graphs. 
In this setting, we find that the improvements due to $\pal$ are negligible, since the ratio of failures which are stop type for $\bp+\ssf$ is small in this case.
It would be interesting to explore larger degree HGP codes and other decoders to find more suitable applications of $\pal$.

In \app{space-time-cost}, we explain our implementation of $\ssf$ in detail, and include some techniques to modestly reduce the time complexity for quantum codes constructed via the product of $(3,4)$ biregular graphs.

\section{Acknowledgements}

\noindent 
We thank  Anthony Leverrier \& Nicolas Delfosse for comments.
LS would like to thank Patrick Hayden and the Stanford Institute for Theoretical Physics for hosting his visit while this research was undertaken.

\noindent AK acknowledges support from the Bloch Postdoctoral Fellowship at Stanford University and NSF grant CCF-1844628.

\bibliographystyle{unsrtabbrev}
\bibliography{references}

\appendix

\section{Proof of Lemma}[Proof of \lemm{line-like}]
\label{sec:lemma-proof}

In this section we prove \lemm{line-like}, which concerns the support of a correction applied by $\ssf$ when the syndrome is line-like.

First, we rule out trivial cases---the vast majority of $Z$ stabilizer generators will not participate.
Flipping qubits that are not themselves supported on the line can only \emph{increase} the syndrome weight.
Even though it is a trivial observation, we formally state this in the following claim.

\begin{claim}
\label{claim:trivial-case}
If $F$ is in the support of a $Z$ stabilizer generator whose local view does not contain $\mathsf{UNSAT}$, then the weight of the syndrome can only increase by flipping $F$.
\end{claim}

Consequently, $\ssf$ will only apply corrections from within stabilizer generators that are adjacent to the line, i.e.\ using stabilizer generators that contain within their support qubits in the line.
This corresponds to $Z$ stabilizer generators of the form $\Gamma(v_1) \times V_2$.
With this observation, we proceed to study only stabilizer generators that contain $\mathsf{UNSAT}$ in the local view.

Next, we ask how the weight of the syndrome changes when a set of qubits $F$ is flipped.
Consider sets $F$ within the support of $(c_1,v_2)$ for $c_1 \in \Gamma(v_1)$.
Let $F \subseteq \supp(c_1,v_2)$ and let $F_V = F \cap V_1 \times V_2$ and $F_C = F \cap C_1 \times C_2$ be the support of the flip on VV and CC qubits respectively.
If we flip the set $F$, parity checks that are adjacent to a single qubit in $F$ will be affected whereas parity checks that are adjacent to two qubits in $F$ will remain unaffected.
Recall that the set of parity checks adjacent to a single qubit of $F$ is called the unique neighborhood of $F$ and is denoted $\Lambda^{(u)}(F)$.

\noindent \textbf{Observation:} If any unsatisfied parity checks overlap the unique neighborhood, they will be flipped; these are supported on $\Lambda^{(u)}(F) \cap \synd$.
On the other hand, after flipping $F$, the following satisfied parity checks will become unsatisfied $\Lambda^{(u)}(F) \setminus \synd$.

\noindent The change in the syndrome is thus
\begin{align}
    |\Lambda^{(u)}(F) \cap \synd| - |\Lambda^{(u)}(F) \setminus \synd|  ~.
\end{align}
The set $F$ can be used to decrease the syndrome if and only if the above quantity is positive.

To illustrate, we use the figure below.
The parity checks lie along the line $\{v_1\} \times C_2$ corresponding to syndromes are cross hatched.
The set $F$ is a candidate set within the support of $(c_1,v_2)$ and the corresponding qubits are colored.
The green and blue squares in the support of $(c_1,v_2)$ are colored green and blue to depict $F_V$ and $F_C$ respectively.
Parity checks in the multi-neighborhood (in red) are flipped twice.
Parity checks in the unique neighborhood (in green/blue) are flipped once.
\begin{center}
    \begin{tikzpicture}
    \fill[fill=green!20] (0,1.5) rectangle (0.5,2);
    \fill[fill=green!20] (0,2) rectangle (0.5,2.5);
    \fill[fill=blue!20] (1,0) rectangle (1.5,0.5);
    \fill[red!20] (1.5,1.5) rectangle (2,2);
    \fill[red!20] (1.5,2) rectangle (2,2.5);
    \fill[red!20] (1,1.5) rectangle (1.5,2);
    \fill[red!20] (1,2) rectangle (1.5,2.5);
    \fill[blue!20] (1,1) rectangle (1.5,1.5);
    \fill[blue!20] (1,2.5) rectangle (1.5,3);
    \fill[green!20] (2,1.5) rectangle (2.5,2);
    \fill[green!20] (2,2) rectangle (2.5,2.5);    
    \node at (0.25,1.75) {$v_1$};
    \fill[pattern=north west lines] (1,1.5) rectangle (1.5,2);
    \fill[pattern=north west lines] (1.5,1.5) rectangle (2,2);
    \foreach \x in {1,1.5,2} {
    \foreach \y in {1,1.5,2,2.5} {
        \draw (\x,\y) rectangle ({\x+0.5},{\y+0.5});
    }
    }
    \foreach \x in {1,1.5,2} {
        \draw (\x,0) rectangle ({\x+0.5},0.5);
    }
    \foreach \y in {1,1.5,2,2.5} {
        \draw (0,\y) rectangle (0.5,{\y+0.5});
    }
    \draw (0.5,0.5) rectangle (0,0);
    
    \node[rotate=-45] at (-0.2,-0.2) {$(c_1,v_2)$};
    \node at (1.75,-0.5) {$\{c_1\} \times \Gamma(v_2)$};
    \node[rotate=90] at (-0.5,2) {$\Gamma(c_1) \times \{v_2\}$};
\end{tikzpicture}
\end{center}

Suppose we restrict our attention to $Z$ stabilizers whose local view contains $\mathsf{UNSAT}$.
If the syndrome $\synd$ is supported on the line $\{v_1\} \times C_2$, qubits that are themselves not on the line will never be flipped by $\ssf$.
\begin{claim}
\label{claim:only-on-line}
    Let $(c_1,v_2) \in C_1 \times V_2$ be a generator whose local view overlaps $\mathsf{UNSAT}$ non-trivially.
    Suppose $F \subseteq \supp(c_1,v_2)$ but $F \cap \{v_1\} \times V_2 = \emptyset$.
    $F$ will not be flipped by $\ssf$.
\end{claim}
\begin{proof}
To evaluate the change in the syndrome, we are only concerned with the restriction of the syndrome to the local view of $(c_1,v_2)$.
This is expressible as $\synd \cap \Gamma(c_1) \times \Gamma(v_2) = \{v_1\} \times \synd_C$ where $\synd_C$ is supported on $\Gamma(v_2)$.

\noindent Recall that we can write $\Lambda^{(u)}(F)$ as a disjoint union
\begin{align}
     F_V \times F_C^c \sqcup F_V^c \times F_C~.
\end{align}
Therefore, we can write $\Lambda^{(u)}(F) \setminus \synd$ in terms of these two portions separately.
Let $|F_V \setminus \{v_1\}| = a$ and $|F_C \setminus \synd_C| = b$ for the sake of readability.
The quantity $a$ expresses the number of VV qubits in the support of $(c_1,v_2)$ that are flipped that are \emph{not} along the line $\{v_1\} \times V_2$.
The quantity $b$ represents CC qubits \emph{not} adjacent to the unsatisfied parity checks.

\noindent We first consider the negative contribution to the change in the syndrome, $|\Lambda^{(u)}(F) \setminus \synd|$.
As $\synd$ is only supported along $\{v_1\} \times \Gamma(v_2)$, the first portion $|F_V \times F_C^c \setminus \synd|$ can be expressed as
\begin{align}
\label{eq:part1}
    &|(F_V \setminus \{v_1\}) \times  F_C^c \setminus \synd_C| \nonumber\\
    =&|F_V \setminus \{v_1\}| \cdot ((\Delta_V - |\synd_C|) - |F_C \setminus \synd_C|)\nonumber\\
    =&a [(\Delta_V - |\synd_C|) - b]~.
\end{align}
The second portion is $|F_V^c \times F_C \setminus \synd|$, which is
\begin{align}
\label{eq:part2}
    &|F_V^c \setminus \{v_1\} \times (F_C \setminus \synd_C)| \nonumber \\
    = &\left( (\Delta_C-1) - |F_V \setminus \{v_1\}| \right)\cdot |F_C \setminus \synd_C| \nonumber \\
    = &[\Delta_C-1 - a] \cdot b~.
\end{align}

\noindent Together, $|\Lambda^{(u)}(F) \setminus \synd|$ is the sum of the contributions from Eq.\ \eqref{eq:part1} and Eq.\ \eqref{eq:part2}:
\begin{align}
    & (\Delta_V - |\synd_C|)a + (\Delta_C - 1)b -2ab\\
    =& (\Delta_V - |\synd_C \union F_C|)a + (\Delta_C - |F_V \union \{v_1\}|)b~.
\end{align}

\noindent Next, we consider the positive contribution to the change in the syndrome, $|\synd \setminus \Lambda^{(u)}(F)|$
\begin{align}
     (F_V \cap \{v_1\}) \times (\synd_C \setminus F_C^c) \sqcup (F_V^c \cap \{v_1\}) \times (\synd_C \setminus F_C).
\end{align}
Observe that $v_1$ can either belong to either $F_V$ or $F_V^c$; it cannot belong to both.
As a result, we compute the change in the syndrome case-by-case.

\noindent \textbf{Case 1:} $F_V \cap \{v_1\} = \{v_1\}$.
Noting that $|\synd_C \setminus F_C^c| = |\synd_C \cap F_C|$, the change in the syndrome is
\begin{align}
\label{eq:case1}
    =&|F_C| - \left(\Delta_V - |\synd_C \union F_C|\right)a - \left(\Delta_C-|F_V|\right)b
\end{align}
Observe that $|\synd_C \union F_C| \leq \Delta_V$ and $|F_V| \leq \Delta_C$.
Therefore, Eq.\ ~\eqref{eq:case1} is clearly maximum when $a = b = 0$.

\noindent \textbf{Case 2:} $F_V^c \cap \{v_1\} = \{v_1\}$.
The change in the syndrome is
\begin{align}
    |\synd_C \setminus F_C| - (\Delta_V - |\synd_C \union F_C|)a  - (\Delta_C-1-|F_V|)b~.
\end{align}
Again, $|\synd_C \union F_C| \leq \Delta_V$ and $|F_V|+1\leq \Delta_C$.
This is maximum when $a = 0$, $b=0$ and $F_C \cap \synd_C = \emptyset$.
However, this is simply the trivial set $F = \emptyset$.
Equivalently, we can let $|F_V| = \Gamma(c_1) \setminus \{v_1\}$ and $F_C \sqcup \synd_C = \Gamma(v_2)$.

The score is normalized by the weight of $F$, and therefore, we pick the smallest set $F$ that results in the largest change in the syndrome.
The only valid corrections, therefore, correspond to flipping the qubit $(v_1,v_2)$ in the support of $(c_1,v_2)$.
\end{proof}

\noindent Claim~\ref{claim:trivial-case} and Claim~\ref{claim:only-on-line} together imply Lemma~\ref{lem:line-like}.

\section[PAL implementation details]{$\pal$ implementation details}
\label{app:details}
We use the implementation of $\osd$ from Roffe \emph{et al}.\ \cite{roffe2020decoding} and we use the variant $\osd$-$\ttt{CS}$.
For $\bp$, we use the min-sum variant of $\bp$ and fix the maximum number of iterations to $30$.
This is in contrast to \cite{roffe2020decoding} where $\bp$ runs for a number of iterations equal to the block length and is instead along the lines of the decoder proposed by Panteleev and Kalachev \cite{panteleev2021degenerate}.

To initialize $\bp$, we estimate the LLRs as a function of the inital error rate $\pphys$.
To this end, we use the following heuristic.
When $\ssf$ results in a stopping failure and the stopping set is line like, we record the average number of bit errors in that line.
Of course, the final decoder does not have access to the actual distribution of errors; rather, we use this heuristic to collect data and inform our choice of LLRs.
We refer to the average number of errors as $q_{V}$ and $q_C$ for line-like errors in $\cC$ and $\tC$ respectively.
These are both functions of $\pphys$ and need to be measured accordingly.
Inferring $q_V$ and $q_C$ is only done once and therefore does not affect the time complexity.

\section[Using PAL after BP + SSF]{Using $\pal$ after $\bp + \ssf$}
\label{app:bp-predec}
Using a combination $\ssf$ along with other decoding algorithms has been considered in the literature.
A popular approach has been to use belief propagation ($\bp$) followed by an appropriate post-decoding algorithm such as $\osd$.
Similarly, $\ssf$ can also be used as a low-complexity decoder to complement $\bp$ \cite{Grospellier2021}.
In this section, we consider applying $\pal$ as a post decoder for $\bp + \ssf$, forming a decoder we will call $\bp + \ssf + \pal$.

Before we specify details of $\bp + \ssf + \pal$, we first note that a number of different options have been considered for $\bp + \ssf$.
The implementations in \cite{Grospellier2021} used a sum-product version of $\bp$ prior to running $\ssf$.
However, the implementations in \cite{panteleev2021degenerate,roffe2020decoding} used a min-sum version of $\bp$ prior to running $\ssf$, which can lead to orders-of-magnitude improvement in the logical error probability.

The implementation in in \cite{Grospellier2021}, which uses a sum-product version of $\bp$, stops at regular intervals to see whether $\ssf$ can correct the error.
If it cannot, it continues running $\bp$ before trying $\ssf$ again.
This poses some challenges since this sum-product version of $\bp$ often does not terminate.
For instance, the total syndrome weight associated with the deduced flip oscillates with the number of iterations of $\bp$.
To handle these challenges, the user specifies the maximum number of iterations of $\bp$ to try before declaring failure.

\begin{figure*}
    \centering
    \begin{tikzpicture}
        \node at (0,0) {\includegraphics[scale=0.6]{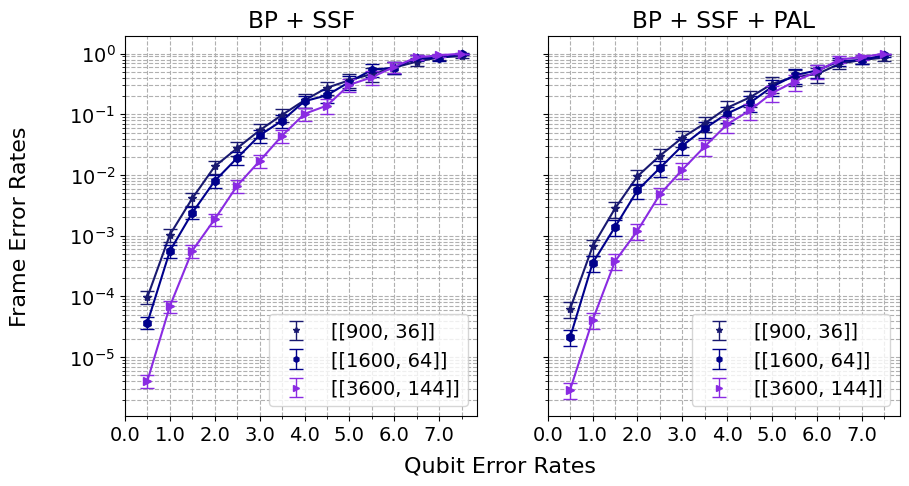}};
        \fill[white] (-6.8,-1.5) rectangle (-6.25,2);
        \node[rotate=90] at (-6.5,0) {$\mathsf{Logical}$ $\mathsf{error}$ $\mathsf{rate}$};
    \end{tikzpicture}
    \caption{Comparing the change in logical error probability before and after the post-processing algorithm $\ttt{PAL}$ while using $\bp$ as a pre-processing algorithm.
    Hypergraph product of $(3, 4)$-biregular bipartite graphs ($\Delta_q = 8$ and $\Delta_g = 7$) are used.
     The error bars represent the 95\% confidence intervals, i.e $\approx$ 1.96 standard deviation.
    }
    \label{fig:all34}
\end{figure*}

On the other hand, we observed that the min-sum version $\bp$ (similar to that in \cite{panteleev2021degenerate,roffe2020decoding}) did not exhibit the oscillatory behavior which was observed in the sum-product case.
We therefore implement the min-sum version of $\bp$ with a fixed number of rounds prior to $\ssf$ as specified in Algorithm \ref{alg:BPSSFNonIt}).
Moreover, the logical error rates achieved when the min-sum version of $\bp$ is used are considerably lower than those in \cite{Grospellier2021} with the sum-product version.
Furthermore, for the range of values that we simulate for, we did not observe an error floor.
We refer to \fig{all34}.

\begin{algorithm}[h]
    \begin{algorithmic}[0]\onehalfspacing
        \State \textbf{Input:} $R \in \bbN$ and $\pphys \in [0,1]$
        \State \textbf{Output:} Correction $\Ehat$.
    \end{algorithmic}
    \begin{algorithmic}[1]\onehalfspacing
        \State Initialize $\bp$ with LLRs equal to $\lambda_q = \log_2{\left((1-\pphys)/\pphys\right)}$.
        \For{$1\leq i \leq R$}
        \State min-sum $\bp$.
        \EndFor
        \For{$q \in Q$}
            \If{ $\lambda_q > 0$}
            \State $\Ehat \leftarrow \Ehat \union \{q\}$
            \EndIf
        \EndFor
        \State $\synd \leftarrow \synd(\Ehat)$
        \State $\Ehat \leftarrow \ssf(\synd)$.
        \State \Return $\Ehat$
    \end{algorithmic}
    \caption{$\bp+\ssf(R, \pphys)$}
    \label{alg:BPSSFNonIt}
\end{algorithm}

In \fig{all34}, we examine the impact of running $\pal$ after this min-sum based $\bp + \ssf$ decoder with $(3,4)$ regular codes.
Using $\pal$ after $\bp+\ssf(R,\pphys)$, we only observe negligible improvements for this code family because the failures were not dominated by stop failures in this case.
It would be interesting to look at other code families (such as codes with larger degree) and other decoders, and observe if $\pal$ offers more significant improvements.

\section{Improving time costs of the SSF decoder}
\label{app:space-time-cost}
In this section, we discuss the cost of na\"ively implementing $\ssf$.
We then demonstrate how some of these costs can be improved.
For simplicity, we only present the time cost associated with decoding $Z$ errors.
The total time cost will be twice of what we present below once we include $X$ errors.

In Algorithm \ref{alg:ssflip-detailed}, we present an implementation of $\ssf$ with details regarding the associated data structures.

\noindent \textbf{Overview:}
 As we explain further below, $\ttt{lookup\_table}$ is simply a table\footnote{$\ttt{lookup\_table}$ is implemented with a heap. Thus ``sorting'' is actually done implicitly at insertion time.} of size $M_Z$ that stores $\genscore(c,v)$ for every generator $(c,v) \in C \times V$.
For each generator $(c,v) \in C \times V$, we have a data structure $\ttt{generator}(c,v)$ that is used to compute the score and store the best subset of qubits to flip in each round.

We first sort $\ttt{lookup\_table}$ to arrange generators in decreasing order of the best score.
We then find the best generator, denoted $(c_i,v_i)$, using the $\ttt{find\_best()}$ method of $\ttt{lookup\_table}$.

The corresponding set of qubits that maximizes the score is denoted $F_i$ and stored in $\ttt{generator}(c_i,v_i)$.
It is obtained by using $\mathtt{generator}(c_i,v_i).\mathtt{best}\_\mathtt{flip}$.
After flipping $F_i$, we recompute scores of the generators that have changed, and reinsert them in $\ttt{lookup\_table}$.
As only a constant number of qubits $F_i$ have been flipped, we only need reinsert a constant number of generators into $\ttt{lookup\_table}$.

\begin{algorithm}
	\begin{algorithmic}[0]
		\State \textbf{Input:} A syndrome $\synd \in \bbF_2^{M}$.
		\State \textbf{Output:} Correction $\Ehat$
    \end{algorithmic}
    \begin{algorithmic}[1]\onehalfspacing
        \For{every $Z$-type generator $(c,v) \in C\times V$} 
        \State $\mathtt{max}$-$\genscore \leftarrow 0$.
            \For{$F \in \supp(c,v)$} \Comment{$2^{\Delta_V + \Delta_C}$ iterations}
                \If{$\genscore(F) > \mathtt{max}$-$\genscore$}
                    \State $\ttt{lookup\_table}(c,v) \leftarrow \genscore(F)$
                \EndIf
            \EndFor
        \EndFor
        \State $\ttt{lookup\_table.sort()}$ \Comment{$O(M \cdot \log(M))$ time}
        \State $\synd_X^{(0)} \leftarrow \synd_X$
		\State $\Ehat \leftarrow \emptyset$
        \State $(c_i,v_i) \leftarrow \ttt{lookup\_table.find\_best()}$ \Comment{$O(1)$ time}
        \State $F_i \leftarrow \mathtt{generator}(c_i,v_i).\ttt{best\_flip()}$.
		\While{$\genscore(F_i) > 0$}
            \State $\Ehat \leftarrow \Ehat \symmdiff F_i \pmod{2}$
            \State $\ttt{updated}$-$\ttt{parity}$-$\ttt{checks} \leftarrow \emptyset$
            \For{Each $X$-type parity check $(\nu,\zeta)$ adjacent to $F_i$}
                \State $\synd_X^{(i+1)}[\nu,\zeta] = \synd_X^{(i)}[\nu,\zeta] \oplus \synd_X(F_i) \pmod{2}$
                \State $\ttt{updated}$-$\ttt{parity}$-$\ttt{checks} \leftarrow \ttt{updated}$-$\ttt{parity}$-$\ttt{checks} \cup (\nu,\zeta)$
            \EndFor
            \For{Each $Z$-type generator $(c,v)$ adjacent to $\ttt{updated}$-$\ttt{parity}$-$\ttt{checks}$}
             \State $\mathtt{max}$-$\genscore \leftarrow 0$.
                \For{$F \in \supp(c,v)$} \Comment{$2^{\Delta_V + \Delta_C}$ iterations}
                    \If{$\genscore(F) > \mathtt{max}$-$\genscore$}
                        \State $\ttt{lookup\_table}(c,v) \leftarrow \genscore(F)$ 
                    \EndIf
                \EndFor
            \EndFor
            \State $i \leftarrow i+1$
            \State $(c_i,v_i) = \ttt{lookup\_table.find\_best()}$ \Comment{$O(1)$ time}
            \State $F_i \leftarrow \mathtt{generator}(c_i,v_i).\mathtt{best}\_\mathtt{flip()}$.
		\EndWhile
		\State \Return $\Ehat$.
	\end{algorithmic}
	\caption{$\ssf(\synd_X$) (Formal)}
	\label{alg:ssflip-detailed}
\end{algorithm}

\subsection{Time complexity}

In this section, we discuss the time costs associated with $\ssf$.
We can form a simple model of the run time of $\ssf$ on a fixed syndrome $\synd_X$ for a quantum expander code of length $N$ encoding $K$ logical qubits as follows:
\begin{eqnarray}
T_{\ssf} &\simeq& T_{\text{init}} + N_{\text{rounds}}\cdot T_{\text{update\_per\_round}}~.
\end{eqnarray}
In words, we can split the contribution to the cost of $\ssf$ from the time to initialize the algorithm, and then the time to perform all iterations.
This can be written as a product of the number of rounds, $N_{\text{rounds}}$, times the time per round $T_{\text{update\_per\_round}}$.

\textbf{1. The initialization costs $T_{\text{init}}$:}
Before each round, the algorithm needs to know the score of each generator.
This then allows it to decide which generator contains the support of the correction to be applied in that round.
The score of all $M_Z$ $Z$-type generators are computed at the start.
These scores are stored in the data structure $\ttt{lookup\_table}$ whose space and time costs we describe shortly.
In our model, we suppose that the subroutine $\ttt{score()}$ which calculates the score of a generator  takes a time $T_{\text{gen-score}}$.
This cycles through the power set of the support and therefore takes time $2^{\Delta_V + \Delta_C}$.
This is a method for the $\ttt{generator}$ class that we also describe below.
Tallying costs, creating the table $\ttt{lookup\_table}$ takes time $N_{\text{generators}}\cdot T_{\text{gen-score}}$.

\textbf{2. Number of rounds $N_{\text{rounds}}$:}
Each round of $\ssf$ strictly reduces the syndrome weight, which is initially $|\synd_X|$. 
Therefore the number of rounds $N_\text{rounds}$ is between zero and $|\synd_X|$, and in our model we assume that it scales proportionally to $|\synd_X|$, i.e., $N_\text{rounds} = c |\synd_X|$ for some constant $0 < c \leq 1$.

\textbf{3. Update time $T_{\text{update\_per\_round}}$:}
We can break down $T_{\text{update\_per\_round}}$ as follows:
\begin{align}
    T_{\text{update\_per\_round}} = \nonumber\\
    &T_{\text{find\_best}} \nonumber\\
    + &T_{\text{update\_deduced\_error}} \nonumber \\
    + &T_{\text{update\_syndromes}} \nonumber \\
    + &T_{\text{update\_scores}}~.
\end{align}
\textbf{3a: $T_{\text{find\_best}}$:}
In each iteration, we flip qubits in the support of the generator with the best score.
Once we have decided to flip the corresponding set of qubits, we need to update the syndromes of the $X$-type parity checks.

The score for each generator is stored in the data structure $\ttt{lookup\_table}$.
The time complexity of $\ssf$ depends on the complexity of querying and updating the data structure $\ttt{lookup\_table}$.
This data structure is unspecified by the algorithm itself.
It is often stated that $\ssf$ is linear time, and what is meant is that the data structure $\ttt{lookup\_table}$ only needs to be queried and updated $O(N)$ times.
Regardless of how the data structure is designed, there is a fundamental trade-off between the time to find the best element using $\ttt{find\_best}$ (line 6) and the time to update elements in each iteration (lines 11-13).

We first describe the data structures $\ttt{lookup\_table}$, the methods associated with it and the corresponding space and time complexity.
This is summarized in Table \ref{tab:datastructs}.

\begin{table*}
    \centering
    \begin{tabular}{|l|c|c|}
        \hline
        \texttt{Method} & Space cost & Time cost\\
        \hline
        {\raggedleft \ttt{init}} & $\Theta(M)$ & $\Theta(M)$ \\
        - \ttt{score()} & $\Theta(1)$ & $2^{\Delta_V + \Delta_C}$\\
        - \ttt{insert} &  $\Theta(1)$  & $\Theta(\log(M))$\\
        \ttt{find\_best()} & $\Theta(1)$ & $\Theta(1)$\\
        \hline
    \end{tabular}\qquad
     \begin{tabular}{|l|c|c|}
        \hline
        \texttt{Method} & Space cost & Time cost\\
        \hline
        \ttt{init} & $\Theta(1)$ & $\Theta(1)$ \\
        \ttt{score()} & $\Theta(1)$ & $\Theta(2^{\Delta_V + \Delta_C})$\\
        \ttt{best\_flip()} &  $\Theta(\Delta_V + \Delta_C)$ & {$\Theta(1)$}\\
        \hline
    \end{tabular}
    \caption{On the left: Data structure \ttt{lookup\_table}, implemented via a heap.
    $\ttt{score()}$ and $\ttt{insert}$ are subroutines of $\ttt{init}$.
    The total time for $\ttt{init}$ is $\Theta[(2^{\Delta_V + \Delta_C} + \log(M))M]$.
    On the right: Data structure \ttt{generator} $(c,v) \in C \times V$.}
    \label{tab:datastructs}
\end{table*}
We use a heap to build the data structure $\ttt{lookup\_table}$.
The time to construct the table is $\Theta(M \log M)$. The time to insert a new element is $\Theta(\log M)$ and the time to find an element is $\Theta(1)$.

\textbf{3b: $T_{\text{update\_deduced\_error}}$:}
This step is simple: we update the vector $\Ehat$ in at most $\Delta_V + \Delta_C$ locations.

\textbf{3c: $T_{\text{update\_syndromes}}$:}
For each of the at most $\Delta_V + \Delta_C$ qubits, we need to update the syndromes of at most $\Delta_V \cdot \Delta_C$ $X$-type parity checks.

\textbf{3d: $T_{\text{update\_scores}}$:}
Each element of $\ttt{lookup\_table}$ is computed using the \ttt{generator} data structure which is described in Table \ref{tab:datastructs}.

The generator data structure stores the best flip corresponding to the generator.
This requires $\Delta_V + \Delta_C$ memory, one bit per qubit in the support of the generator.
As the $\ttt{score()}$ method iterates through every possible flip for qubits in the support of the $Z$-type generator, it also stores the value of the current set of qubits in consideration and the corresponding score.

\subsection{Reducing the run time to score a generator}
\label{sec:graycode}

 The ideas in \sec{graycode} are due to Antoine Grospellier and were used in \cite{grospellier2018numerical}.
We include it here for completeness.

The most costly part of a na\"ive implementation of $\ssf$ is the function $\ttt{score\_gen()}$.
To illustrate, we include the following data from a profiler in Table \ref{tab:gray-before-after}.
It is evident that speedups in $\ttt{score\_gen}$ would mean significant speedups overall.
Grospellier's idea was to use a \emph{Gray} code \cite{doran2007gray}.

To compute the score for a $Z$ generator $(c,v) \in C \times V$, we need to cycle through each subset of $\Gamma_Z(c,v)$.
There are $2^{\Delta_V + \Delta_C}$ such elements which can be indexed by an integer $i \in \{1,...,2^{\Delta_V+\Delta_C}\}$.
To this end, we could use a lexicographic ordering on bits.
The lexicographic ordering of bits is the bijection $\ttt{lex}: \{0,1\}^{\Delta} \leftrightarrow \bbZ$, such that for $\{0,1\}^{\Delta} \ni \textbf{x} = (x_1,...,x_{\Delta})$:
\[
    (x_1,x_2,\cdots, x_{\Delta}) \leftrightarrow  2^{\Delta-1}x_1 + \cdots + 2 x_{\Delta-1} + x_{\Delta} ~.
\]
The obvious way to compute the score might be to use an enumeration algorithm: for $0 \leq x \leq 2^{\Delta_V}$ and $0 \leq y \leq 2^{\Delta_C}$.
The local view $N(c,v)$ is the matrix of syndromes of parity checks adjacent to $(c,v)$, i.e.\
\begin{equation}
    N(c,v) = \begin{pmatrix}
        \synd(\nu_1, \zeta_1) & \cdots & \synd(\nu_1, \zeta_{\Delta_V})\\
        \vdots            &        &   \vdots \\
        \synd(\nu_{\Delta_C}, \zeta_1) & \cdots & \synd(\nu_{\Delta_C}, \zeta_{\Delta_V})
    \end{pmatrix}
\end{equation}
$\genscore$ is a function of the difference in syndromes within this local view.
If we let $\ttt{lex}(x)$ and $\ttt{lex}(y)$, the change in the syndrome by flipping $x$ and $y$ is represented as
\begin{equation}
    \ttt{lex}(x) \cdot N(c,v) \cdot \ttt{lex}(y) - N(c,v)~.
\end{equation}

However, this obvious ordering is wasteful as it performs computations multiple times.
For instance, suppose we have computed the difference in syndrome for some $x$ and $y$.
If we proceed in lexicographic order, the next computation would be
\begin{align}
    \ttt{lex}(x) \cdot N(c,v) \cdot \ttt{lex}(y+1)~.
\end{align}
First, observe that $\ttt{lex}(x) \cdot N(c,v)$ is fixed and need not be recomputed.
We need only keep track of the \emph{change} in the score.
However, the trouble is that $\ttt{lex}(y)$ and $\ttt{lex}(y+1)$ could differ on many locations, i.e.\ the Hamming weight between the two can be quite large.
Although this sounds trivial, this actually leads to many redundant computations.
It would be useful if $\ttt{lex}(y)$ and $\ttt{lex}(y+1)$ differed in only one location as this means that only a single element of $\ttt{lex}(x)\cdot N(c,v)$ changes.
The map $\ttt{lex}$ by itself does not possess this property.

\begin{table}[h]
\centering
\begin{tabular}{|c|c|c|}
\hline
    $x$ & $\ttt{lex}(x)$ & $\ttt{gray}(x)$\\
\hline
    0  & 00 & 00 \\
    1  & 01 & 01 \\
    2  & 10 & 11\\
    3  & 11 & 10 \\
\hline
\end{tabular}
\caption{An example of a Gray code mapping as compared to lexographic ordering for $x \in \{0,...,3\}$.
Successive integers only differ in one location in the Gray code mapping.}
\label{tab:gray}
\end{table}
A Gray code is a total order on binary strings that is different from the lexicographic order.
It allows for a different mapping from bits to integers such that only a single element changes between rounds.
As an example, we contrast the gray code on 2 bits below in Table \ref{tab:gray}.

This leads to a savings in the time spent on computing $\genscore$.
We include an illustrative table below for the $(3,4)$ regular codes below in Table \ref{tab:gray-before-after}.
For the $\dsl 10000, 400\dsr$ code, there is approximately a $10\times$ speedup.
Before this optimization, computing the scores in each round is roughly $90\%$ of the total run time.
It drops to $55\%$ after using the Gray code.

\begin{table*}
    \centering
    \begin{tabular}{|c|c|c|c|c|}
        \hline
        Algorithm      &  number          & Avg.\ run time & Avg.\ run time     & Avg.\ run time per sample\\
                      & score gen calls  & per call (s)   & for all calls (s)  & (s)\\
        \hline
        \hline
        \ttt{$\ssf_0$} &  $1988$ & $2.1 \times 10^{-5}$ & $4.2 \times 10^{-2}$ & $4.6 \times 10^{-2}$ \\
        \hline
        \ttt{$\ssf_1$}  &  $1990$ & $1.7 \times 10^{-6}$ & $3.4 \times 10^{-3}$ & $6.1 \times 10^{-3}$ \\
        \hline
    \end{tabular}
    \caption{Comparison of the costs and error correction performance of an unoptimized implementation $\ssf_0$ and an implementation of $\ssf$, $\ssf_1$, with the Gray code for length 10000 HGP code with $p=0.01$ error rate.}
    \label{tab:gray-before-after}
\end{table*}

\subsection[Reducing the number of calls of the score subroutine]{Reducing the number of calls of the \ttt{score} subroutine}

The na\"ive algorithm is computing the $\genscore$ for generators that \emph{certainly} cannot be the best.
In particular, we can ignore a generator if it is not adjacent to an unsatisfied $X$-type parity check as these generators cannot possibly reduce the score.
Table~\ref{tab:ignore-score-gen} demonstrates the performance benefits by using this observation.

\begin{table}[h]
    \centering
    \resizebox{\columnwidth}{!}{%
    \begin{tabular}{|c|c|c|c|c|}
        \hline
        Algorithm      &  number          & Avg.\ run time \\
                       & score gen calls  & (s)  \\
        \hline
        $\ssf_1$ &  $1988$ & $6.1 \times 10^{-3}$  \\
        \hline
        $\ssf_2$     &  $1369$ & $4.4 \times 10^{-3}$  \\
        \hline
    \end{tabular}%
    }
    \caption{Comparison of the costs and error correction performance of $\ssf_1$ that uses the Gray code and an optimized implementation of SSF, $\ssf_2$, that avoids generators not adjacent to unsatisfied parity checks for a $\dsl 10000, 400\dsr$ code with $\Delta_V=3, \Delta_C=4$ and $p=0.01$ error rate.}
    \label{tab:ignore-score-gen}
\end{table}

For each generator, computing the score takes $\Theta(2^{\Delta_V + \Delta_C})$ time.
This can be reduced if we can restrict the weight of the flip $F$ to size $w$.
In this case, the time complexity would scale in proportion to $\binom{\Delta_V + \Delta_C}{w}$ which can be significantly less than $2^{\Delta_V + \Delta_C}$.
Furthermore, if \emph{most} flips had low weight, it could significantly improve the time spent computing the score.
Theoretical results \cite{leverrier2015quantum,fawzi2018efficient,fawzi2018constant} cannot guide us here.
They only prove that in each iteration of $\ssf$, there exists a generator within whose support we can find \emph{some} set of qubits to reduce the error weight.

Consider a quantum expander code generated from the product of $(3,4)$-biregular graphs $G$.
For a generator $(c,v) \in C \times V$, let $N(c,v)$ denote the local view of the $Z$-type generator and $|N(c,v)|$ denote the weight of the generator, i.e.\ the number of unsatisfied parity checks in the local view.

For $i \in \{1,...,4\}$, let $r_{ba}(i)$ denote the row weight of the $i$\textsuperscript{th} row of $N(c,v)$.
Similarly, for $j \in \{1,...,3\}$, let $c_{ba}(j)$ denote the column weight of the $j$\textsuperscript{th} row of $N(c,v)$.
We call this the $\ttt{Waterfall}$ algorithm.
\begin{algorithm}[h]
\begin{algorithmic}[1]\onehalfspacing
\If{$\exists i \in \{1,...,4\}$ such that $r_{ba}(i) \geq 2$}
    \State $\genscore(c,v) \leftarrow |3-2r_{ba}(i)|$
    \State $F^* \leftarrow \{i\}$
\ElsIf{$\exists j \in \{1,2,3\}$ such that $c_{ba}(j) = 3$}
    \State $\genscore(c,v) \leftarrow |4-2c_{ba}(j)|$
    \State $F^* \leftarrow \{j\}$
\ElsIf{$\exists j \in \{1,2,3\}$ such that $c_{ba}(j) = 2$ $\wedge$ $r_{ba}(i) = 1$}
    \State $\genscore(c,v) = 0.5$
    \State $F^* \leftarrow \{i,j\}$
\Else
    \State $\genscore(c,v) \leftarrow 0$
    \State $F^* \leftarrow \emptyset$
\EndIf
\end{algorithmic}
\caption{$\ttt{Waterfall}$}
\end{algorithm}

The following claim shows that the algorithm works as expected.
\begin{claim}
\label{claim:waterfall}
Let $G$, $N(c,v)$, $r_{ba}$ and $c_{ba}$ be defined as above.
If $|N(c,v)| \leq 3$, it suffices to run Algorithm $\mathtt{Waterfall}$.
\end{claim}
\begin{proof}
    When the local view satisfies $|N(c,v)| \leq 3$, there are very few cases to test owing to equivalence up to row and column permutations.
    This means that \emph{locally}, i.e.\ from the point-of-view of a generator $(c,v) \in C \times V$, the number of rows and columns flipped is only a function of the weight of the local view $N(c,v)$.
    Each of these scenarios is listed below.
    
   \noindent \textbf{Case 1:} $|N(c,v)| = 1$.

   \noindent In this case, there is no valid flip.
   To prove this, it is sufficient to assume that the $(1,1)$ entry of the matrix $N(c,v)$ is unsatisfied.

    \noindent \textbf{Case 2:} $|N(c,v)| = 2$.

    \noindent If the two unsatisfied parity checks are not in the same row, then, there is no valid flip.
    To prove this, it is sufficient to consider the following: assume that the $(1,1)$ entry is unsatisfied and
    \begin{enumerate}
        \item $(2,1)$ entry is unsatisfied. This shows that if two unsatisfied parity checks are adjacent to the same CC qubit, then no flip exists.
        \item $(2,2)$ entry is unsatisfied. This shows that if two unsatisfied parity checks are not adjacent to the same VV or CC qubit, then no flip exists.
    \end{enumerate}
    If they are in the same row $i \in \{1,...,4\}$, then we flip that row.

    \noindent \textbf{Case 3:} $|N(c,v)| = 3$.

    \noindent If the three unsatisfied parity checks are such that no two of them are in the same row or column, then there is no valid flip.
    There is only one case to consider: let the $(1,1)$, $(2,2)$ and $(3,3)$ entries of $N(c,v)$ be unsatisfied.

    \noindent There are many ways we can find a valid flip.
    \begin{enumerate}
        \item If there are two or more unsatisfied parity checks in a single row $i \in \{1,...,4\}$, then flip row $i$.
        \item If there are two unsatisfied parity checks in one column $j \in \{1,2,3\}$ and a row $i$ such that the $(i,j_i)$ entry is unsatisfied where $j_i \neq j$, then we flip both row $i$ and column $j$.
    \end{enumerate}
\end{proof}

\begin{table}
    \centering
    \resizebox{\columnwidth}{!}{%
    \begin{tabular}{|l|c|c|c|}
    \hline
        Algorithm                      & Calls to & Calls to & Time \\
        &$\ttt{score}\_\ttt{gen}$ & $\ttt{Waterfall}$ & (s)\\
    \hline
    \hline
        $\ssf_2$                    & $5,508$ & $0$ & $9.8 \times 10^{-3}$ \\
        $\ssf_3$ & $220$ & $5,268$ & $7.0 \times 10^{-3}$  \\
    \hline
    \end{tabular}%
    }
    \caption{Number of average calls to $\ttt{score}\_\ttt{gen}$ before and after $\ttt{Waterfall}$ for a code with $10000$ physical qubits $(3, 4)$ codes with $p = 0.01$.
    $\ssf_2$ merely ignores generators not adjacent to unsatisfied parity checks whereas $\ssf_3$ utilizes $\ttt{Waterfall}$.}
    \label{tab:post-waterfall}
\end{table}

\noindent The time improvements offered by $\ttt{Waterfall}$ are shown in Table~\ref{tab:post-waterfall}.
We find a modest improvement in the total time required.

\end{document}